\documentclass{article}
\usepackage{fullpage}
\usepackage{amsthm,amssymb,amsmath}  
\usepackage{xspace,enumerate}
\usepackage[utf8]{inputenc}
\usepackage{thmtools}
\usepackage{thm-restate}
\usepackage{authblk}
\usepackage{todonotes}
\usepackage{hyperref}
\usepackage[noline,noend,linesnumbered,algo2e,ruled]{algorithm2e}
\usepackage{caption}
\usepackage{subcaption}
\usepackage{comment}
\usepackage{wrapfig}
\usepackage{url}
  \theoremstyle{plain}
  \newtheorem{theorem}{Theorem}
  \newtheorem{lemma}{Lemma}  
  \newtheorem{corollary}[theorem]{Corollary}

  \theoremstyle{definition}
  \newtheorem{definition}{Definition}
  
  \newtheorem{example}{Example}

\newcommand{\TR}{\textsf{TR}}
\newcommand{\CT}{\textsf{CT}}
\newcommand{\ST}{\textsf{ST}}
\newcommand{\A}{\mathcal{L}}
\newcommand{\M}{\mathcal{M}}
\newcommand{\T}{\mathcal{T}}
\newcommand{\G}{\mathcal{A}}
\newcommand{\RANGE}{\mathcal{R}}
\newcommand{\R}{\mathcal{S}}
\newcommand{\cO}{\mathcal{O}}

\newcommand\DNA {\textsc{DNA}\xspace}
\newcommand\DBLP {\textsc{XML}\xspace}
\newcommand\ENG {\textsc{ENGLISH}\xspace}
\newcommand\PRO {\textsc{PROTEINS}\xspace}
\newcommand\SOURCES {\textsc{SOURCES}\xspace}
\newcommand\STD {\textsc{STD}\xspace}
\newcommand\WIN {\textsc{WIN}\xspace}
\newcommand\SYN {\textsc{SYN}\xspace}
\newcommand\kddDBLP {\textsc{DBLP}\xspace}
\newcommand\UNIREF {\textsc{UNIREF}\xspace}

\def\dd{\mathinner{.\,.}}  
\title{String Sampling with Bidirectional String Anchors}

\author[1]{Grigorios Loukides}
\author[2,3]{Solon P.\ Pissis}
\author[2]{Michelle Sweering}

\affil[1]{Department of Informatics, King's College London, London, UK\protect\\ \texttt{grigorios.loukides@kcl.ac.uk}}
\affil[2]{CWI, Amsterdam, The Netherlands, \texttt{[solon.pissis,michelle.sweering]@cwi.nl}}
\affil[3]{Vrije Universiteit, Amsterdam, The Netherlands}

\date{\today}

\begin{document}

\maketitle

\begin{abstract}
The minimizers sampling mechanism is a popular mechanism for string sampling introduced independently by Schleimer et al.~[SIGMOD 2003] and by Roberts et al.~[\emph{Bioinf.}~2004]. Given two positive integers $w$ and $k$, it selects the lexicographically smallest length-$k$ substring in every fragment of $w$ consecutive length-$k$ substrings (in every sliding  window of length $w+k-1$). Minimizers samples are approximately uniform, locally consistent, and computable in linear time. Although they do not have good worst-case guarantees on their size, they are often small in practice. They thus have been successfully employed in several string processing applications. Two main disadvantages of minimizers sampling mechanisms are: first, they also do not have good guarantees on the expected size of their samples for every combination of $w$ and $k$; and, second, indexes that are constructed over their samples do not have good worst-case guarantees for on-line pattern searches.

To alleviate these disadvantages, we introduce bidirectional string anchors (bd-anchors), a new string sampling mechanism. Given a positive integer $\ell$, our mechanism selects the lexicographically smallest rotation in every length-$\ell$ fragment (in every sliding  window of length $\ell$). We show that bd-anchors samples are also approximately uniform, locally consistent, and computable in linear time. In addition, our experiments using several datasets demonstrate that the bd-anchors sample sizes decrease proportionally to $\ell$; and that these sizes are competitive to or smaller than the minimizers sample sizes using the analogous sampling parameters. We provide theoretical justification for these results by analyzing the expected size of bd-anchors samples. As a negative result, we show that computing a total order $\leq$ on the input alphabet, which minimizes the bd-anchors sample size, is NP-hard.

We also show that by using any bd-anchors sample, we can construct, in near-linear time, an index  which requires linear (extra) space in the size of the sample and answers on-line pattern searches in near-optimal time. We further show, using several datasets, that a simple implementation of our index is consistently faster for on-line pattern searches than an analogous implementation of a minimizers-based index [Grabowski and Raniszewski, \emph{Softw.~Pract.~Exp.}~2017].

Finally, we highlight the applicability of bd-anchors by developing an efficient and effective heuristic for top-$K$ similarity search under edit distance. We  show, using synthetic datasets, that our heuristic is more accurate and more than one order of magnitude faster in top-$K$ similarity searches than the state-of-the-art tool for the same purpose [Zhang and Zhang, KDD 2020].
\end{abstract}

\section{Introduction}

The notion of \emph{minimizers}, introduced independently by Schleimer et al.~\cite{DBLP:conf/sigmod/SchleimerWA03} and by Roberts et al.~\cite{DBLP:journals/bioinformatics/RobertsHHMY04}, is a mechanism to sample a set of positions over an input string. The goal of this sampling mechanism is, given a string $T$ of length $n$ over an alphabet $\Sigma$ of size $\sigma$, to simultaneously satisfy the following properties: 
\begin{description}
    \item[Property 1 (approximately uniform sampling):] Every sufficiently long fragment of $T$ has a representative position sampled by the mechanism. 
    \item[Property 2 (local consistency):] Exact matches between sufficiently long fragments of $T$ are preserved unconditionally by having the same (relative) representative positions sampled by the mechanism.
\end{description}
In most practical scenarios, sampling the smallest number of positions is desirable, as long as Properties 1 and 2 are satisfied.  This is because it leads to small data structures or fewer computations. 
Indeed, the minimizers sampling mechanism satisfies the property of approximately uniform sampling: given two positive integers $w$ and $k$, it selects at least one length-$k$ substring in every fragment of $w$ consecutive length-$k$ substrings (Property 1). 
Specifically, this is achieved by selecting the starting positions of the smallest length-$k$ substrings in every $(w+k-1)$-long fragment, where smallest is defined by a choice of a total order on the universe of length-$k$ strings. These positions are called the ``minimizers''. 
Thus from similar fragments, similar length-$k$ substrings are sampled (Property 2). In particular, if two strings have a fragment of length $w+k-1$ in common, then they have at least one minimizer corresponding to the same length-$k$ substring. Let us denote by $\mathcal{M}_{w,k}(T)$ the set of minimizers of string $T$. The following example illustrates the sampling.

\begin{example}
The set $\mathcal{M}_{w,k}$ of minimizers for   $w=k=3$ for string $T=\texttt{aabaaabcbda}$ (using a 1-based index) is $\mathcal{M}_{3,3}(T)=\{1,4,5,6,7\}$ and for string $Q=\texttt{abaaa}$ is  $\mathcal{M}_{3,3}(Q)=\{3\}$.
Indeed $Q$ occurs at position $2$ in $T$; and $Q$ and $T[2\dd 6]$ have the minimizers $3$ and $4$, respectively, which both correspond to string \texttt{aaa} of length $k=3$.
\end{example}

 The minimizers sampling mechanism is very versatile, and it has been employed in various ways in many different applications~\cite{10.1093/bioinformatics/btw152,Kraken14,DBLP:journals/bioinformatics/DeorowiczKGD15,DBLP:journals/bioinformatics/ChikhiLM16,DBLP:journals/spe/GrabowskiR17,DBLP:journals/jcb/JainDKAP18,DBLP:journals/bioinformatics/JainKDPA18,DBLP:journals/bioinformatics/Li18,DBLP:journals/bioinformatics/JainRZCWKP20}. Since its inception, the minimizers sampling mechanism has undergone numerous theoretical and practical improvements~\cite{DBLP:conf/wabi/OrensteinPMSK16,DBLP:journals/bioinformatics/ChikhiLM16,DBLP:journals/bioinformatics/MarcaisPBOSK17,DBLP:journals/bioinformatics/MarcaisDK18,DBLP:conf/bcb/DeBlasioGKM19,DBLP:conf/recomb/EkimBO20,DBLP:journals/bioinformatics/ZhengKM20,DBLP:journals/bioinformatics/JainRZCWKP20,Zheng2021} with a particular focus on minimizing the size of the residual sample; see Section~\ref{sec:related} for a summary on this line of research. Although minimizers have been extensively and successfully used, especially in bioinformatics, we observe several inherent problems with setting the parameters $w$ and $k$. In particular, although the notion of length-$k$ substrings (known as \emph{$k$-mers} or \emph{$k$-grams}) is a widely-used string processing tool, we argue that, in the context of minimizers, it may be causing many more problems than it solves: it is not clear to us why one should use an extra sampling parameter $k$ to effectively characterize a fragment of length $\ell=w+k-1$ of $T$. In what follows, we describe some problems that may arise when setting the parameters $w$ and $k$.
\begin{description}
\item[Indexing:] The most widely-used approach is to index the selected minimizers using a hash table.
The \emph{key} is the selected length-$k$ substring and the \emph{value} is the list of positions it occurs.
If one would like to use length-$k'$ substrings for the minimizers with 
$\ell=w+k-1=w'+k'-1$, for some $w'\neq w$ and $k'\neq k$, they should compute the new set $\mathcal{M}_{w',k'}(T)$ of minimizers and construct their new index based on $\mathcal{M}_{w',k'}$ from scratch. 
    \item[Querying:] To the best of our knowledge, no index based on minimizers can return in optimal or near-optimal time all occurrences of a pattern $Q$ of length $|Q|\geq \ell=w+k-1$ in $T$. 
    \item[Sample Size:] If one would like to minimize the number of selected minimizers, they should consider different total orders on the universe of length-$k$ strings, which may complicate practical implementations, often scaling only up to a small $k$ value, e.g.~$k=16$~\cite{DBLP:conf/recomb/EkimBO20}.
    On the other hand, when $k$ is fixed and $w$ increases, the length-$k$ substrings in a fragment become increasingly decoupled from each other, and that \emph{regardless of the total order} we may choose. Unfortunately, this interplay phenomenon is inherent to minimizers. It is known that $k\geq \log_\sigma (w)+c$, for a fixed constant $c$, is a \emph{necessary condition} for the existence of minimizers samples with expected size in $\cO(n/w)$~\cite{DBLP:journals/bioinformatics/ZhengKM20}; see Section~\ref{sec:related}.
\end{description}

We propose the notion of bidirectional string anchors (bd-anchors) to alleviate these disadvantages. The bd-anchors is a mechanism that drops the sampling parameter $k$ and its corresponding disadvantages. We only fix a parameter $\ell$, which can be viewed as the length $w+k-1$ of the fragments in the minimizers sampling mechanism. The \emph{bd-anchor} of a string $X$ of length $\ell$ is the lexicographically smallest rotation (cyclic shift) of $X$. We unambiguously characterize this rotation by its leftmost starting position in string $XX$. The set $\G_{\ell}(T)$ of the order-$\ell$ bd-anchors of string $T$ is the set of bd-anchors of all length-$\ell$ fragments of $T$. It can be readily verified that bd-anchors satisfy Properties 1 and 2.

\begin{example}
The set $\G_{\ell}(T)$ of bd-anchors for $\ell=5$ for string $T=\texttt{aabaaabcbda}$ (using a 1-based index) is $\G_{5}(T)=\{4,5,6,11\}$ and for string $Q=\texttt{abaaa}$, $\G_{5}(Q)=\{3\}$.
Indeed $Q$ occurs at position $2$ in $T$; and $Q$ and $T[2\dd 6]$ have the bd-anchors $3$ and $4$, respectively, which both correspond to the rotation \texttt{aaaab}. 
\end{example}

Let us remark that \emph{string synchronizing} sets, introduced by Kempa and Kociumaka~\cite{DBLP:conf/stoc/KempaK19}, is another string sampling mechanism which may be employed to resolve the disadvantages of minimizers. Yet, it appears to be quite complicated to be efficient in practice. For instance, in~\cite{dinklage_et_al:LIPIcs:2020:12905}, the authors used a simplified and specific definition of string synchronizing sets to design a space-efficient data structure for answering longest common extension queries.

We consider the word RAM model of computations with $w$-bit machine words, where $w=\Omega(\log n)$, for stating our results. We also assume throughout that string $T$ is over alphabet $\Sigma=\{1,2,\ldots,n^{\cO(1)}\}$, which captures virtually any real-world scenario. We measure space in terms of $w$-bit machine words. We make the following three specific contributions:

\begin{enumerate}
    \item In Section~\ref{sec:bd-anchors} we show that the set $\G_{\ell}(T)$, for any $\ell>0$ and any $T$ of length $n$, can be constructed in $\cO(n)$ time. We generalize this result showing that   
    for any constant $\epsilon \in(0,1]$, $\G_{\ell}(T)$ can be constructed in $\cO(n + n^{1-\epsilon}\ell)$ time using $\cO(n^\epsilon+\ell+|\G_{\ell}|)$ space.
    Furthermore, we show that the expected size of $\G_{\ell}$ for strings of length $n$, randomly generated by a memoryless source with identical letter probabilities, is in $\cO(n/\ell)$, for any integer $\ell>0$. The latter is in contrast to minimizers which achieve the expected bound of $\cO(n/w)$ only when $k\geq \log_\sigma w + c$, for some constant $c$~\cite{DBLP:journals/bioinformatics/ZhengKM20}.
    We then show, using five real datasets, that indeed the size of $\G_{\ell}$ decreases proportionally to $\ell$; that it is competitive to or smaller than $\mathcal{M}_{w,k}$, when $\ell=w+k-1$; and that it is \emph{much smaller} than $\mathcal{M}_{w,k}$ for \emph{small} $w$ values, which is practically important, as widely-used aligners that are based on minimizers will require less space and computation time if bd-anchors are used instead. Finally, we show a negative result using a reduction from minimum feedback arc set: computing a total order $\leq$ on $\Sigma$ which minimizes $|\G_{\ell}(T)|$ is NP-hard.
    \item In Section~\ref{sec:index} we show an index based on $\G_{\ell}(T)$, for any string $T$ of length $n$ and any integer $\ell>0$, which answers on-line pattern searches in near-optimal time. In particular, for any constant $\epsilon>0$, we show that our index supports the following  space/query-time trade-offs:
\begin{itemize}
    \item it occupies $\cO(|\G_{\ell}(T)|)$ extra space and reports all $k$ occurrences of any pattern $Q$ of length $|Q|\geq \ell$ given on-line in $\cO(|Q|+(k+1)\log^{\epsilon}(|\G_{\ell}(T)|))$ time; or
    \item it occupies $\cO(|\G_{\ell}(T)|\log^{\epsilon}(|\G_{\ell}(T)|))$ extra space and reports all $k$ occurrences of any pattern $Q$ of length $|Q|\geq \ell$ given on-line in $\cO(|Q|+\log\log(|\G_{\ell}(T)|)+k)$ time.
\end{itemize}
We also show that our index can be constructed in $\cO(n +|\G_{\ell}(T)|\sqrt{\log(|\G_{\ell}(T)|)})$ time.
We then show, using five real datasets, that a simple implementation of our index is \emph{consistently faster} in on-line pattern searches than an analogous implementation of the minimizers-based index proposed by Grabowski and Raniszewski in~\cite{DBLP:journals/spe/GrabowskiR17}. 
\item In Section~\ref{sec:edit} we highlight the applicability of bd-anchors by developing an efficient and effective heuristic for top-$K$ similarity search under edit distance. This is a fundamental and extensively studied problem~\cite{DBLP:conf/vldb/KahveciS01,DBLP:conf/sigmod/ChaudhuriGGM03,DBLP:conf/stoc/ColeGL04,DBLP:conf/vldb/LiWY07,DBLP:conf/aaai/YangYK10,DBLP:conf/sigmod/ZhangHOS10,DBLP:journals/tods/Qin0XLLW13,DBLP:journals/pvldb/WangDTZ13,DBLP:conf/sigmod/WangLF12,DBLP:conf/sigmod/DengLF14,DBLP:journals/tkde/HuLBFWGX16,DBLP:journals/vldb/YuWLZDF17,DBLP:conf/kdd/Zhang020} with applications in areas including bioinformatics, databases, data mining, and information retrieval. We show, using synthetic datasets, that our heuristic, which is based on the bd-anchors index, is \emph{more accurate} and \emph{more than one order of magnitude faster} in top-$K$ similarity searches than the state-of-the-art tool proposed by Zhang and Zhang in~\cite{DBLP:conf/kdd/Zhang020}. 
\end{enumerate}

In Section~\ref{sec:prel}, we provide some preliminaries; and in Section~\ref{sec:related} we discuss works related to minimizers. Let us stress that, although other works may be related to our contributions, we focus on comparing to minimizers because they are extensively used in applications. The source code of our implementations is available at~\url{https://github.com/solonas13/bd-anchors}. 

A preliminary version of this paper appeared as~\cite{DBLP:conf/esa/LoukidesP21}.

\section{Preliminaries}\label{sec:prel}

We start with some basic definitions and notation following~\cite{DBLP:books/daglib/0020103}.
An \emph{alphabet} $\Sigma$ is a finite nonempty set of elements called \emph{letters}.
A \emph{string} $X=X[1]\ldots X[n]$ is a sequence of \emph{length} $|X|=n$ of letters from $\Sigma$. The \emph{empty} string, denoted by $\varepsilon$, is the string of length $0$.
The fragment $X[i\dd j]$ of $X$ is an \emph{occurrence} of the underlying \emph{substring} $S=X[i]\ldots X[j]$. 
We also say that $S$ occurs at \emph{position} $i$ in $X$. 
A {\em prefix} of $X$ is a fragment of $X$ of the form $X[1\dd j]$ and a {\em suffix} of $X$ is a fragment of $X$ of the form $X[i\dd n]$. 
The set of all strings over $\Sigma$ (including $\varepsilon$) is denoted by $\Sigma^*$. The set of all length-$k$ strings over $\Sigma$ is denoted by $\Sigma^k$. Given two strings $X$ and $Y$, the {\em edit distance} $d_\mathrm{E}(X,Y)$ is the minimum number of edit operations (letter insertion, deletion, or substitution) transforming one string into the other. 

Let $M$ be a finite nonempty set of strings over $\Sigma$ of total length $m$. We define the {\em trie} of $M$, denoted by $\TR(M)$, as a deterministic finite automaton that recognizes $M$. Its set of states (nodes) is the set of prefixes of the elements of $M$; the initial state (root node) is $\varepsilon$; the set of terminal states (leaf nodes) is $M$; and edges are of the form $(u,\alpha,u\alpha)$, where $u$ and $u\alpha$ are nodes and $\alpha\in\Sigma$. The size of $\TR(M)$ is thus $\cO(m)$.
The \emph{compacted trie} of $M$, denoted by $\CT(M)$, contains the root node, the branching nodes, and the leaf nodes of $\TR(M)$. The term compacted refers to the fact that $\CT(M)$ reduces the number of nodes by replacing each maximal branchless path segment with a single edge, and that it uses a fragment of a string $s\in M$ to represent the label of this edge in $\cO(1)$ machine words. The size of $\CT(M)$ is thus $\cO(|M|)$. When $M$ is the set of suffixes of a string $Y$, then $\CT(M)$ is called the \emph{suffix tree} of $Y$, and we denote it by $\ST(Y)$. The suffix tree of a string of length $n$ over an alphabet $\Sigma=\{1,\ldots,n^{\cO(1)}\}$ can be constructed in $\cO(n)$ time~\cite{DBLP:conf/focs/Farach97}. 

Let us fix throughout a string $T=T[1\dd n]$ of length $|T|=n$ over an ordered alphabet $\Sigma$. Recall that we make the standard assumption of an integer alphabet $\Sigma=\{1,2,\ldots,n^{\cO(1)}\}$. 

We start by defining the notion of minimizers of $T$ from~\cite{DBLP:journals/bioinformatics/RobertsHHMY04} (the definition in~\cite{DBLP:conf/sigmod/SchleimerWA03} is slightly different). Given an integer $k>0$, an integer $w>0$, and the $i$th length-$(w+k-1)$ fragment $F=T[i\dd i+w+k-2]$ of $T$, we define the \emph{$(w,k)$-minimizers} 
of $F$ as the positions $j\in[i,i+w)$ where a lexicographically minimal length-$k$ substring of $F$ occurs. The set $\M_{w,k}(T)$ of $(w,k)$-minimizers of $T$ is defined as the set of $(w,k)$-minimizers of $T[i\dd i+w+k-2]$, for all $i\in[1,n-w-k+2]$. The \emph{density} of $\mathcal{M}_{w,k}(T)$ is defined as the quantity $|\mathcal{M}_{w,k}(T)|/n$.
The following bounds are obtained trivially. The density of any minimizer scheme is at least $1/w$, since at least one $(w,k)$-minimizer is selected in each fragment, and at most $1$, when every $(w,k)$-minimizer is selected. 

If we waive the lexicographic order assumption, the set $\mathcal{M}_{w,k}(T)$ can be computed on-line in $\cO(n)$ time, and if we further assume a constant-time computable function that gives us an \emph{arbitrary} rank for each length-$k$ substring in $\Sigma^k$ in constant amortized time~\cite{DBLP:journals/bioinformatics/JainRZCWKP20}. This can be implemented, for instance, using a rolling hash function (e.g.~Karp-Rabin fingerprints~\cite{DBLP:journals/ibmrd/KarpR87}), and the rank (total order) is defined by this function. We also provide here, for completeness, a simple off-line $\cO(n)$-time algorithm that uses a lexicographic order. 

\begin{theorem}\label{the:min}
The set $\mathcal{M}_{w,k}(T)$, for any integers $w,k>0$ and any string $T$ of length $n$, can be constructed in $\cO(n)$ time. 
\end{theorem}

\begin{proof}
The underlying algorithm has two main steps. In the first step, we construct $\ST(T)$, the suffix tree  of $T$ in $\cO(n)$ time~\cite{DBLP:conf/focs/Farach97}.
Using a depth-first search traversal of $\ST(T)$ we assign at every position of $T$ in $[1,n-k+1]$ the lexicographic rank of $T[i\dd i+k-1]$ among all the length-$k$ strings occurring in $T$. This process clearly takes $\cO(n)$ time as $\ST(T)$ is an ordered structure; it yields an array $R$ of size $n-k+1$ with lexicographic ranks. In the second step, we apply a folklore algorithm, which computes the minimum elements in a sliding window of size $w$ (cf.~\cite{DBLP:journals/bioinformatics/JainRZCWKP20}) over $R$. The set of reported indices is $\mathcal{M}_{w,k}(T)$. 
\end{proof}

\section{Bidirectional String Anchors}\label{sec:bd-anchors}

We introduce the notion of bidirectional string anchors (bd-anchors). 
Given a string $W$, a string $R$ is a \emph{rotation} (or cyclic shift or conjugate) of $W$ if and only if there exists a decomposition $W=UV$ such that $R=VU$, for a string $U$ and a nonempty string $V$. We often characterize $R$ by its starting position $|U|+1$ in $WW=UVUV$. We use the term rotation interchangeably to refer to string $R$ or to its identifier $(|U|+1)$.

\begin{definition}[Bidirectional anchor]
Given a string $X$ of length $\ell>0$, the \emph{bidirectional anchor} (bd-anchor) of $X$ is the lexicographically minimal rotation $j\in[1,\ell]$ of $X$ with minimal $j$. The set of order-$\ell$ bd-anchors of a string $T$ of length $n>\ell$, for some integer $\ell>0$, is defined as the set $\G_{\ell}(T)$ of bd-anchors of $T[i\dd i+\ell-1]$, for all $i\in [1,n-\ell+1]$. 
\end{definition}

The \emph{density} of $\G_{\ell}(T)$ is defined as the quantity $|\G_{\ell}(T)|/n$.
It can be readily verified that the bd-anchors sampling mechanism satisfies Properties 1 (approximately uniform sampling) and 2 (local consistency).  

\begin{example}
Let $\ell=5$, $T=\texttt{aabaaabcbda}$, and  $T'=\texttt{aacaaaccbda}$. Strings $T$ and $T'$ (are at Hamming distance 2 but) have the same set of bd-anchors of order $5$: $\G_5(T)=\G_5(T')=\{4,5,6,11\}$. The reader can probably share the intuition that the bd-anchors sampling mechanism is suitable for sequence comparison due to Properties 1 and 2, in particular, when the parameter $\ell$ is set accordingly.
\end{example}

\paragraph{Linear-Time Construction of $\G_{\ell}$.}~Importantly, we show that $\G_{\ell}$ admits an efficient construction. One can use the linear-time algorithm by Booth~\cite{DBLP:journals/ipl/Booth80} to compute the lexicographically minimal rotation for each length-$\ell$ fragment of $T$, resulting in an $\cO(n\ell)$-time algorithm, which is reasonably fast for modest $\ell$.  (Booth's algorithm gives the leftmost minimal rotation by construction.) We instead give an optimal $\cO(n)$-time algorithm for the construction of $\G_{\ell}$, which is mostly of theoretical interest.

For every string $X$ and every natural number $m$, we
define the $m$th \emph{power} of the string $X$, denoted by $X^m$ , by $X^0=\varepsilon$ and $X^k = X^{k-1}X$
for $k = 1, 2,\ldots, m$. A nonempty string is \emph{primitive}, if it is not the power of any other string. Let us state two well-known combinatorial lemmas.

\begin{lemma}[\cite{DBLP:books/daglib/0020103}]\label{lem:prim}
A nonempty string $X$ is primitive if and only if it occurs as a substring in $XX$ only as a prefix and as a suffix.
\end{lemma}

\begin{lemma}[\cite{DBLP:conf/fct/ShyrT77}]\label{lem:primconj}
Let $X=UV$ and $R=VU$, for two strings $U,V$. If $X$ is primitive then $R$ is also primitive.
\end{lemma}

A substring $U$ of a string $X$ is called an \emph{infix} of $X$ if and only if $U=X[i\dd j]$ with $i>1$ and $j<n$.

\begin{lemma}\label{lem:one}
A string $X$ has more than one minimal lexicographic rotation if and only if $X$ is a power of some string.
\end{lemma}
\begin{proof}~

\begin{enumerate}
\item[$(\Rightarrow)$] Let $X=U_1V_1$, and $R=V_1U_1$ be the leftmost minimal lexicographic rotation of $X$. Suppose towards a contradiction that $X$  has another minimal lexicographic rotation but $X$ is primitive.
In particular, there exists $H=V_2U_2=R$, with $X=U_2V_2$ and $|U_1|<|U_2|$.
If $X$ is primitive, then $R$ is also primitive by Lemma~\ref{lem:primconj} but then $RR=V_1U_1V_1U_1$ has $H$ occurring as infix. In particular, in $RR$, $V_2$ is a suffix of the first occurrence of $V_1$ and $U_2$ is a prefix of $U_1V_1$ and thus $H=R$ occurs as infix.  By Lemma~\ref{lem:prim} we obtain a contradiction.
\item[$(\Leftarrow)$] Let $X=UU\cdots U$ and a minimal lexicographic rotation of $X$ be $i\in[1,|X|]$. Then either $i+|U|$ or $i-|U|$ is a minimal lexicographic rotation of $X$.
\end{enumerate}
\end{proof}

\begin{example}[Illustration of Lemma~\ref{lem:one}]
Let $X=\texttt{cbacbacba}$, $R=\texttt{acbacba}\cdot\texttt{cb}$ with $U_1=\texttt{acbacba}$ and $V_1=\texttt{cb}$, and $H=\texttt{acba}\cdot\texttt{cbacb}=R$ with $U_2=\texttt{acba}$ and $V_2=\texttt{cbacb}$. Observe that $H$ occurs as infix (shown as  underlined) in  $RR=\texttt{acb\underline{acbacbacb}acbacb}$ hence $X$ is a power of some string.  
\end{example}

\begin{lemma}\label{lem:hash}
Let $X$ be a string of length $n$ and set $Y=XX\#$, for some letter $\#$ not occurring in $X$ that is the lexicographically maximal letter occurring in $Y$.
Further let $Y[i\dd 2n+1]$ be the lexicographically minimal suffix of $Y$, for some $i\in[1,2n]$.
The leftmost lexicographically minimal rotation of $X$ is $i$.
\end{lemma}
\begin{proof}
First note that $i\in[1,n]$ because $\#$ is the lexicographically maximal letter occurring in $Y$. 

We consider two cases: (i) $X$ is primitive; and (ii) $X$ is power of some string. In the first case, $X$ has one lexicographically minimal rotation by Lemma~\ref{lem:one}, and thus this is $i$.
In the second case, $X$ has more than one lexicographically minimal rotations, but because $X$ is power of some string and $\#$ is the lexicographically maximal letter occurring in $Y$, $i$ is the leftmost lexicographically minimal rotation of $X$.
\end{proof}

We employ the data structure of Kociumaka~\cite[Theorem 20]{DBLP:conf/cpm/Kociumaka16} to obtain the following result.

\begin{theorem}\label{the:con}
The set $\G_{\ell}(T)$, for any $\ell>0$ and any $T$ of length $n$, can be constructed in $\cO(n)$ time. 
\end{theorem}
\begin{proof}

The data structure of Kociumaka~\cite[Theorem 20]{DBLP:conf/cpm/Kociumaka16} gives the minimal lexicographic suffix for any concatenation $Y$ of $k$ arbitrary fragments of a string $S$ in $\cO(k^2)$ time after an $\cO(|S|)$ time preprocessing.  

We set $S=T\#$, for some letter $\#$ that does  not occur in $T$ and is the lexicographically maximal letter occurring in $S$. 
For each fragment $T[i\dd i+\ell-1]$, we compute the minimal lexicographic suffix of string $$Y=S[i\dd i+\ell-1]\cdot S[i\dd i+\ell-1]\cdot S[n+1]=T[i\dd i+\ell-1]\cdot T[i\dd i+\ell-1]\cdot \#,$$ where $k=3$ in $\cO(k^2)=\cO(1)$ time. 
This suffix of $Y$ is the minimal lexicographic rotation by Lemma~\ref{lem:hash}. 
\end{proof}

\paragraph{Space-Efficient Construction of $\G_{\ell}$.} It should be clear that, in the best case, the size of $\G_{\ell}$ is in $\cO(n/\ell)$ and this bound is tight. The construction of \autoref{the:con} requires $\cO(n)$ space. Ideally, we would thus like to compute $\G_{\ell}$ efficiently using (strongly) sublinear space. We generalize \autoref{the:con} to the following result.

\begin{theorem}\label{the:space}
The set $\G_{\ell}(T)$, for any $\ell>0$, any $T$ of length $n$, and any constant $\epsilon \in(0,1]$, can be constructed in $\cO(n + n^{1-\epsilon}\ell)$ time using $\cO(n^\epsilon+\ell+|\G_{\ell}|)$ space. 
\end{theorem}
\begin{proof}
We compute $\mathcal{A}_\ell(T[\lceil n^\epsilon (i - 1)\rceil + 1 \dd \max(\lceil n^\epsilon i\rceil + \ell, n)])$ for all $i \in [1, \lceil n^{1-\epsilon}\rceil]$ using the algorithm from \autoref{the:con} and output their union. For any constant $\epsilon \in(0,1]$, the alphabet size $|\Sigma| = n^{\cO(1)} = (n^\epsilon+\ell)^{\cO(1)}$ is still polynomial in the length $n^\epsilon+\ell$ of the fragments, so computing one such anchor set takes $\cO(n^\epsilon+\ell)$ time and space by \autoref{the:con}. We delete each fragment (and the associated data structure) before processing the subsequent anchor set: it takes $\cO(n + n^{1-\epsilon}\ell)$ time and $\cO(n^\epsilon+\ell)$ additional space to construct $\mathcal{A}_\ell(T)$.
\end{proof}

\paragraph{Expected Size of $\G_{\ell}$.}~We next analyze the expected size of $\G_\ell(T)$. We first show that if $\ell$ grows no faster than the size $\sigma$ of the alphabet, then the expected size of $\G_{\ell}$ is in $\cO(n/\ell)$. Otherwise, if $\ell$ grows faster than $\sigma$, we slightly amend the sampling process to ensure that the expected size of the sample is in $\cO(n/\ell)$.

\begin{lemma}\label{lem:randombound} If $T$ is a string of length $n$, randomly generated by a memoryless source over an alphabet of size $\sigma \geq 2$ with identical letter probabilities, then, for any integer $\ell>0$, the expected size of $\G_{\ell}(T)$ is in $\cO(\frac{n\log\ell}{\ell\log\sigma}+\frac{n}{\ell})$.
\end{lemma}
\begin{proof}
If $\ell = 1$, then $\mathcal{A}_\ell(T) = n$. If $\ell = 2$, then $\mathcal{A}_\ell(T) = 1 + (n-2)(2\sigma^2+1)/3\sigma^2$. 
Now suppose $\ell \geq 3$. We say that $T[i\dd i+\ell-1]$ introduces a new bd-anchor if there exists $j \in [1, \ell]$ such that $j$ is the bd-anchor of $T[i\dd i+\ell-1]$, but $j+k$ is not the bd-anchor of $T[i-k\dd i-k+\ell-1]$ for all $k \in [1, \max(\ell - j, i-1)]$. Let $N_i(T)$ denote the event that $T[i\dd i+\ell-1]$ introduces a new bd-anchor. Since the letters are independent identically distributed, the probability $\mathbb{P}[N_i(T)]$ only depends on and is non-increasing in the number of preceding overlapping length-$k$ substrings. Therefore
$$\mathbb{E}[|\mathcal{A}_\ell(T)|] = \mathbb{P}[N_1(T)]+\cdots+\mathbb{P}[N_{n-\ell + 1}(T)] \leq 1 +(n-1) \mathbb{P}[N_2(T)].$$
Let $p$ be the length of the shortest prefix of the lexicographically minimal rotation of $T[2\dd \ell+1]$ which is strictly smaller than the same length prefix of any other rotation of $T[2\dd \ell+1]$.

Note that
\begin{eqnarray}\mathbb{P}[N_2] &\leq& \mathbb{P}[T[1\dd\ell] \text{ or } T[2\dd \ell+1]\text{ is a power of some string}]\label{power}\\ &&+\ \mathbb{P}[T[1\dd \ell]\text{ is primitive with bd-anchor } 1]\label{bd1}\\
&&+\ \mathbb{P}[T[2\dd \ell+1] \text{ is primitive with bd-anchor} > \ell - 3 \log\ell /\log \sigma]\label{bdend}\\
&&+\ \mathbb{P}[T[2\dd \ell+1] \text{ is primitive and } p \geq 3 \log \ell / \log \sigma]\label{prefixrepeat}
\end{eqnarray}
To bound the probability in (\ref{power}), note that 
$$\mathbb{P}[\text{length-$\ell$ string is a power of some string}] \leq \sum_{d < \ell, d\mid \ell} \sigma^{-(\ell-d)} \leq \sum_{d \leq \ell/2} \sigma^{-(\ell-d)} = \sigma^{1-\ell/2}/(\sigma-1).$$
The probability in (\ref{bd1}) is bounded by $1/\ell$ since each letter of a primitive length-$\ell$ string is equally likely to be the anchor.
Similarly, the probability in (\ref{bdend}) is bounded by $(3\log\ell /\log \sigma + 1) /\ell$. Finally, the probability in (\ref{prefixrepeat}) is bounded by the probability that two prefixes of length $\lceil 3 \log \ell / \log \sigma \rceil$ of rotations of $T[2\dd\ell+1]$ are equal, which is at most $\ell^2 \cdot \sigma^{-3 \log \ell / \log \sigma} = 1 / \ell$.
It follows that
\begin{eqnarray*}\mathbb{P}[N_2] &\leq& 2\sigma^{1-\ell/2}/(\sigma-1)+1/\ell+ (3\log\ell /\log \sigma + 1) /\ell + 1/\ell\\
&=& \mathcal{O}\left(\frac{\log \ell}{\ell \log \sigma} +\frac{1}{\ell}\right)
\end{eqnarray*}
We conclude that for any $\ell > 0$ the expected size of $\G_{\ell}(T)$ is in $\cO(\frac{n\log\ell}{\ell\log\sigma}+\frac{n}{\ell})$.
\end{proof}


We define a reduced version of bd-anchors to ensure that the expected size of the sample is in $\cO(n/\ell)$.
\begin{definition}
Given a string $X$ of length $\ell>0$ and an integer $0\leq r\leq \ell-1$, we define the \emph{reduced bidirectional anchor} of $X$ as the lexicographically minimal rotation $j\in[1,\ell-r]$ of $X$ with minimal $j$. The set of order-$\ell$ reduced bd-anchors of a string $T$ of length $n>\ell$ is defined as the set $\G_{\ell}^{\text{red}}(T)$ of reduced bd-anchors of $T[i\dd i+\ell-1]$, for all $i\in [1,n-\ell+1]$. 
\end{definition}

\begin{lemma}\label{lem:reduced} If $T$ is a string of length $n$, randomly generated by a memoryless source over an alphabet of size $\sigma \geq 2$ with identical letter probabilities, then, for any integer $\ell>0$, the expected size of $\G_{\ell}^{\text{red}}(T)$ with $r=\lceil4 \log\ell /\log \sigma\rceil$ is in $\cO(n/\ell)$.
\end{lemma}
\begin{proof}
If $\ell \in \{1,2\}$, then $\G_{\ell}^{\text{red}}(T) \leq n$.
Now suppose $\ell \geq 3$. Analogously to $N$ in \autoref{lem:randombound}, we denote the event that $T[i\dd i+\ell-1]$ introduces a new reduced bd-anchor by $N^{\text{red}}_i(T)$. Again we find
$$\mathbb{E}\left[\left|\mathcal{A}_\ell^{\text{red}}(T)\right|\right] = \mathbb{P}[N^{\text{red}}_1(T)]+\cdots+\mathbb{P}[N^{\text{red}}_{n-\ell + 1}(T)] \leq 1 +(n-\ell) \mathbb{P}[N^{\text{red}}_2(T)].$$
Let $p^{\text{red}}$ be the length of the shortest prefix of the lexicographically minimal rotation $j_1 \in [1, \ell -r]$ of $T[2\dd \ell+1]$ which is strictly smaller than the same length prefix of any other rotation $j_2 \in [1, \ell -r] \setminus \{j_1\}$.
Using a similar proof to that of \autoref{lem:randombound} we find  that
\begin{eqnarray*}\mathbb{P}[N^{\text{red}}_2(T)] &\leq& \mathbb{P}[T[1\dd\ell] \text{ or } T[2\dd \ell+1]\text{ is a power of some string}]\\ &&+\ \mathbb{P}[T[1\dd \ell]\text{ is primitive with bd-anchor } 1]\\&&+\ \mathbb{P}[T[2\dd \ell+1]\text{ is primitive with bd-anchor } \ell - r+1]\\
&&+\ \mathbb{P}[T[2\dd \ell+1] \text{ is primitive and } p^{\text{red}} \geq r]\\
&\leq& 2\sigma^{1-\ell/2}/(\sigma-1)+1/(\ell-r) + 1/(\ell-r) + \ell^2 / \sigma^r = 2/\ell + o\left(1/\ell\right).
\end{eqnarray*}
We conclude that for any $\ell > 0$ the expected size of $\G^{\text{red}}_{\ell}(T)$ is in $\cO(n/\ell)$.
\end{proof}

In particular, if $\ell=\cO(\sigma)$, we employ the sampling mechanism underlying $\G_{\ell}(T)$, otherwise ($\ell=\Omega(\sigma)$) we employ the sampling mechanism underlying $\G_{\ell}^{\text{red}}(T)$ with $r=\lceil4 \log\ell /\log \sigma\rceil$ to ensure that the expected density of the residual sampling is always in $\cO(n/\ell)$.

Constructing $\G_{\ell}^{\text{red}}(T)$ in $\cO(n)$ time requires a trivial modification in Theorem~\ref{the:con}. For each fragment $T[i\dd i+\ell-1]$, instead of computing the minimal lexicographic suffix of string $$Y=S[i\dd i+\ell-1]\cdot S[i\dd i+\ell-1]\cdot S[n+1]=T[i\dd i+\ell-1]\cdot T[i\dd i+\ell-1]\cdot \#,$$ in $\cO(1)$ time, we compute the minimal lexicographic suffix of string $$Y^{\text{red}}=S[i\dd i+\ell-1]\cdot S[i\dd i+\ell-1-r]\cdot S[n+1]=T[i\dd i+\ell-1]\cdot T[i\dd i+\ell-1-r]\cdot \#,$$ in $\cO(1)$ time. We then directly obtain the trade-off in Theorem~\ref{the:space} for constructing $\G_{\ell}^{\text{red}}(T)$.

\paragraph{Density Evaluation.}~We compare the density of bd-anchors and reduced bd-anchors, denoted by BDA and rBDA, respectively, to the density of minimizers, for different values of $w$ and $k$ such that $\ell=w+k-1$. This is a fair comparison because $\ell=w+k-1$ is the length of the fragments considered by both mechanisms.  We implemented bd-anchors, the standard minimizers mechanism from~\cite{DBLP:journals/bioinformatics/RobertsHHMY04}, and the minimizers mechanism with robust winnowing  from~\cite{DBLP:conf/sigmod/SchleimerWA03}. The standard minimizers and those with robust winnowing are referred to as \STD and \WIN, respectively. 

For bd-anchors, we used Booth's algorithm, which is easy to implement and reasonably fast. For minimizers, we used Karp-Rabin fingerprints~\cite{DBLP:journals/ibmrd/KarpR87}. (Note that such ``random'' minimizers tend to perform \emph{even better} than the ones based on lexicographic total order in terms of density~\cite{DBLP:journals/bioinformatics/ZhengKM20}.) For the reduced version of bd-anchors, we used $r=\lceil 3 \log\ell / \log\sigma \rceil$, because the $r$ value suggested by Lemma~\ref{lem:reduced} is relatively large for the small $\ell$ values tested; e.g.~for $\ell=15$, $\lceil 3 \log\ell / \log\sigma \rceil=8$.   
Throughout, we do not evaluate construction times, as all implementations are reasonably fast, and we make the standard assumption that preprocessing is only required once. We used five string datasets from the popular Pizza \& Chili corpus~\cite{pizza} (see Table~\ref{tab:data} for the datasets characteristics). 
All implementations referred to in this paper have been written in \texttt{C++} and compiled at optimization level \texttt{-O3}. All experiments reported in this paper were conducted using a single core of an AMD Opteron 6386 SE 2.8GHz CPU and 252GB RAM running GNU/Linux.

\begin{table}[!t]
\centering 
\begin{tabular}{|c|| c | c|}
\hline {\bf Dataset}  & {\bf Length}  & {\bf Alphabet} \tabularnewline
~ & $n$ & {\bf size}  $|\Sigma|$ 
\tabularnewline\hline\hline
 \DNA &  200,000,000 & 4  \tabularnewline \hline
  \DBLP  & 200,000,000 & 95 \tabularnewline \hline
  \ENG  & 200,000,000 & 224   \tabularnewline \hline
  \PRO  & 200,000,000 & 27  \tabularnewline \hline
  \SOURCES & 200,000,000 & 229  \tabularnewline \hline
  \end{tabular}
\caption{Datasets characteristics.}
  \label{tab:data}
\end{table} 

\begin{figure}[!t]\centering
    \begin{subfigure}[b]{0.6\textwidth}
    \includegraphics[width=\linewidth]{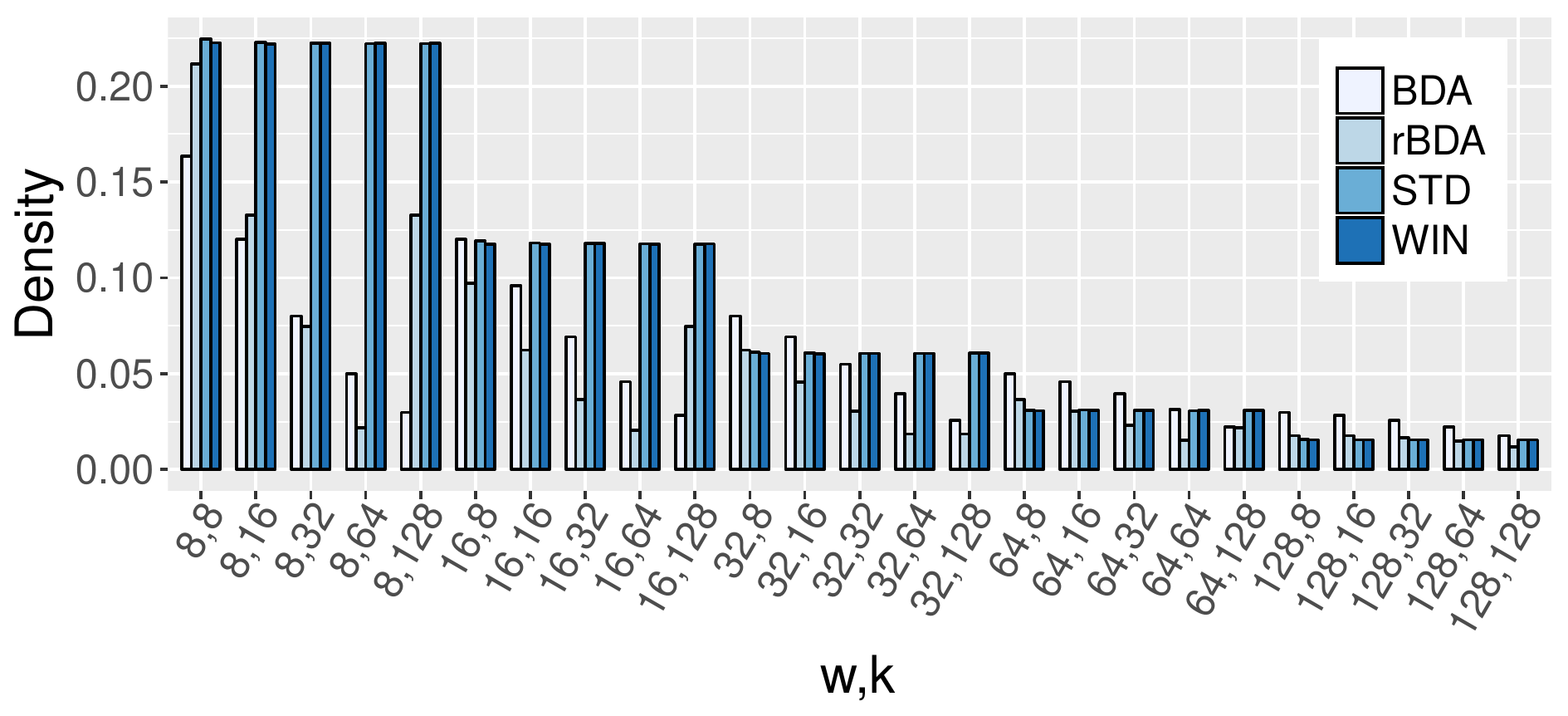}
    \caption{DNA}
    \label{DNA}
    \end{subfigure}\\
    \begin{subfigure}[b]{0.6\textwidth}
    \includegraphics[width=\linewidth]{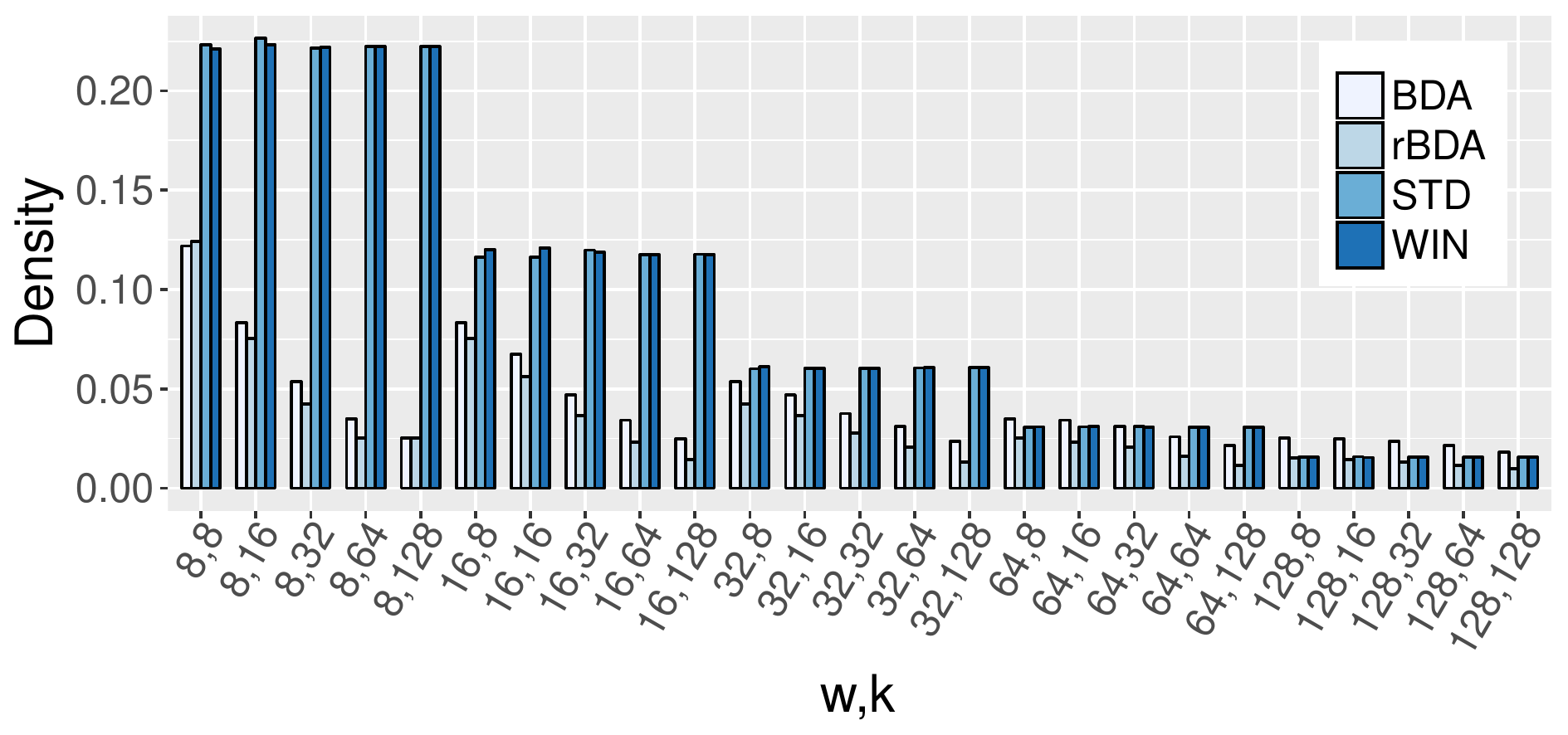}
    \caption{\DBLP}
    \label{DBLP}
    \end{subfigure}\\
        \caption{Density vs.~$w,k$ for $\ell=w+k-1$ and the first two datasets of Table~\ref{tab:data}.}\label{fig:experiment1a}
\end{figure}

\begin{figure}\centering
\begin{subfigure}[b]{0.6\textwidth}
    \includegraphics[width=\linewidth]{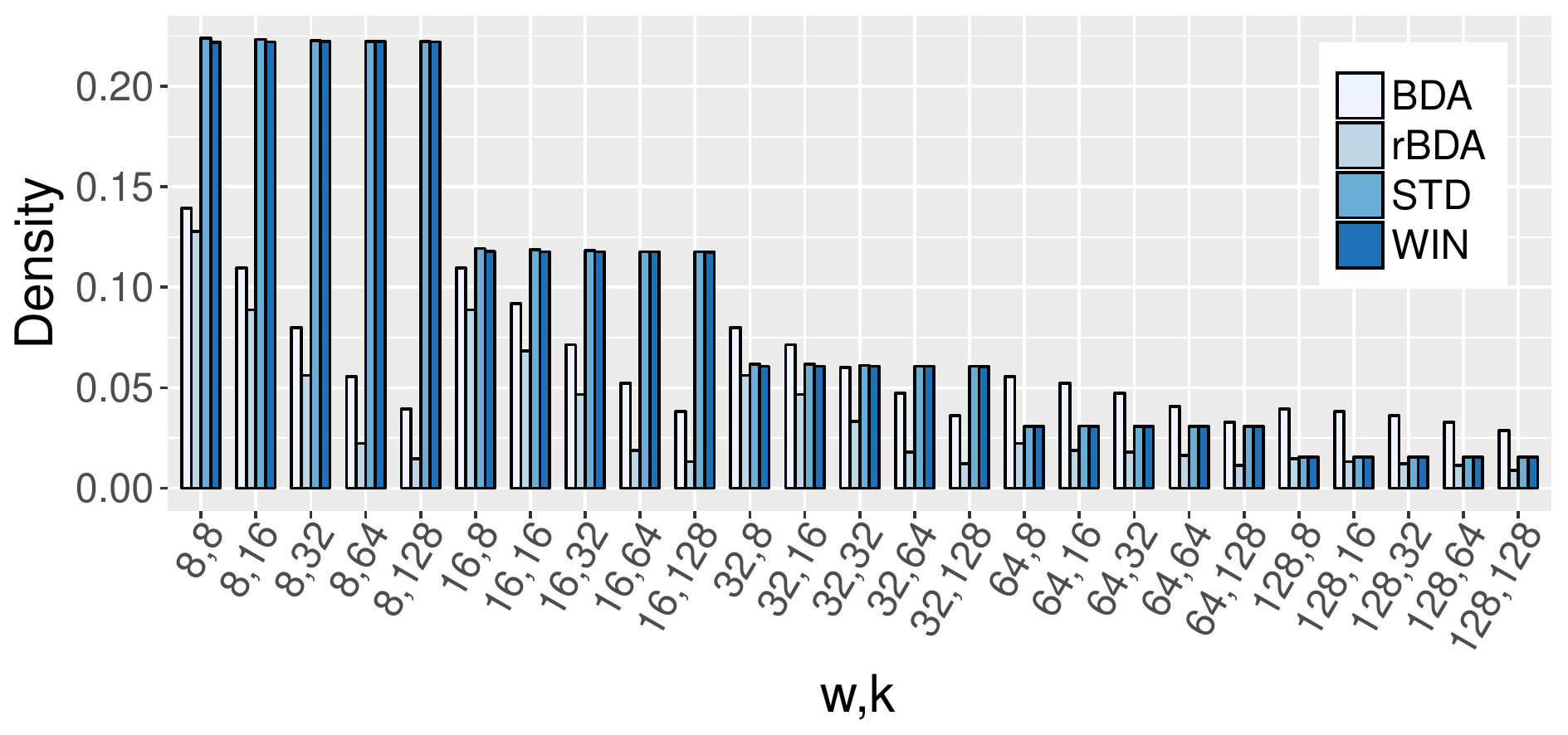}
    \caption{ENGLISH}\label{ENG}
    \end{subfigure}\\
    
    \hspace{+0.1mm}
    \begin{subfigure}[b]{0.6\textwidth}
    \includegraphics[width=\linewidth]{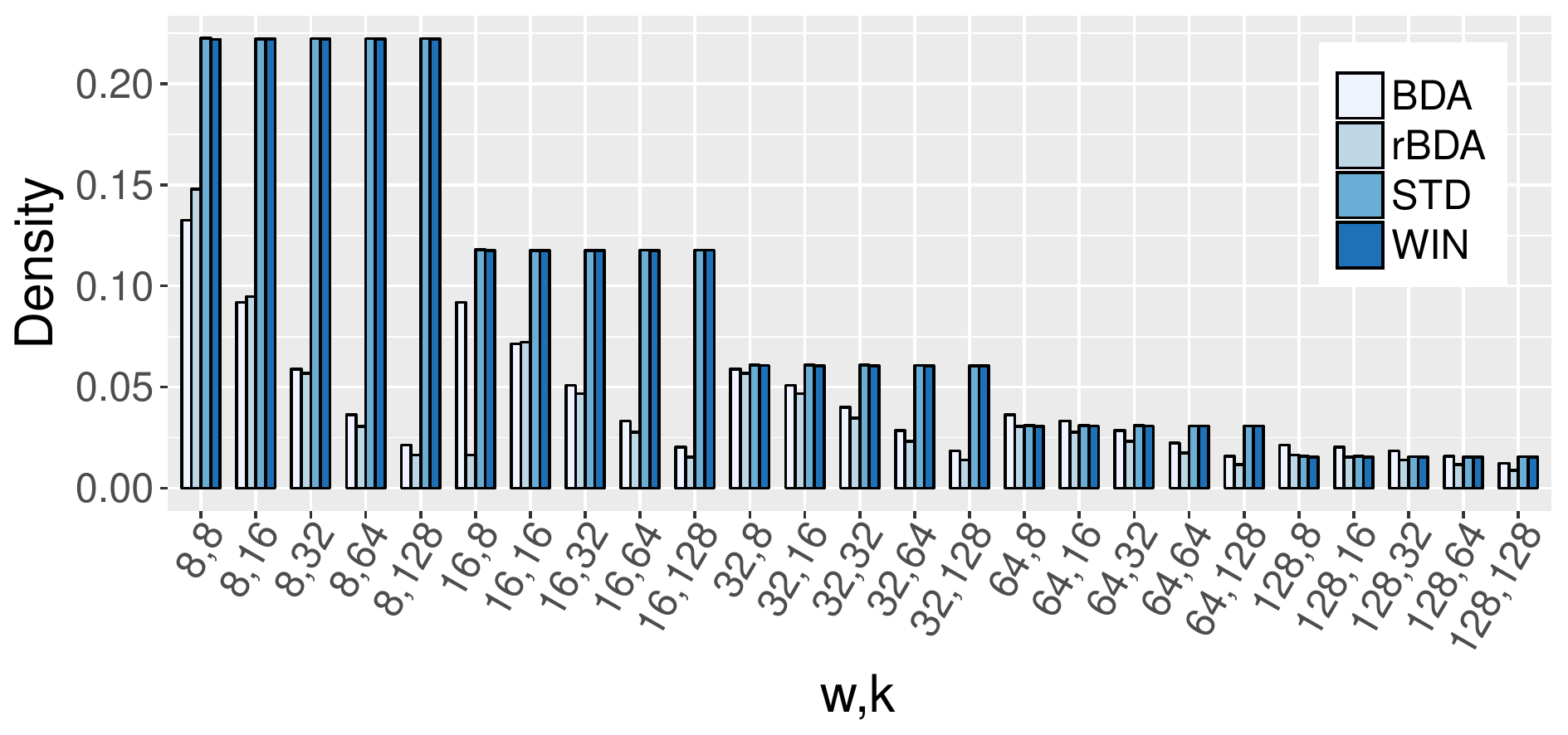}
    \caption{PROTEINS}\label{PROT}
    \end{subfigure}
    \\
    \begin{subfigure}[b]{0.6\textwidth}
    \includegraphics[width=\linewidth]{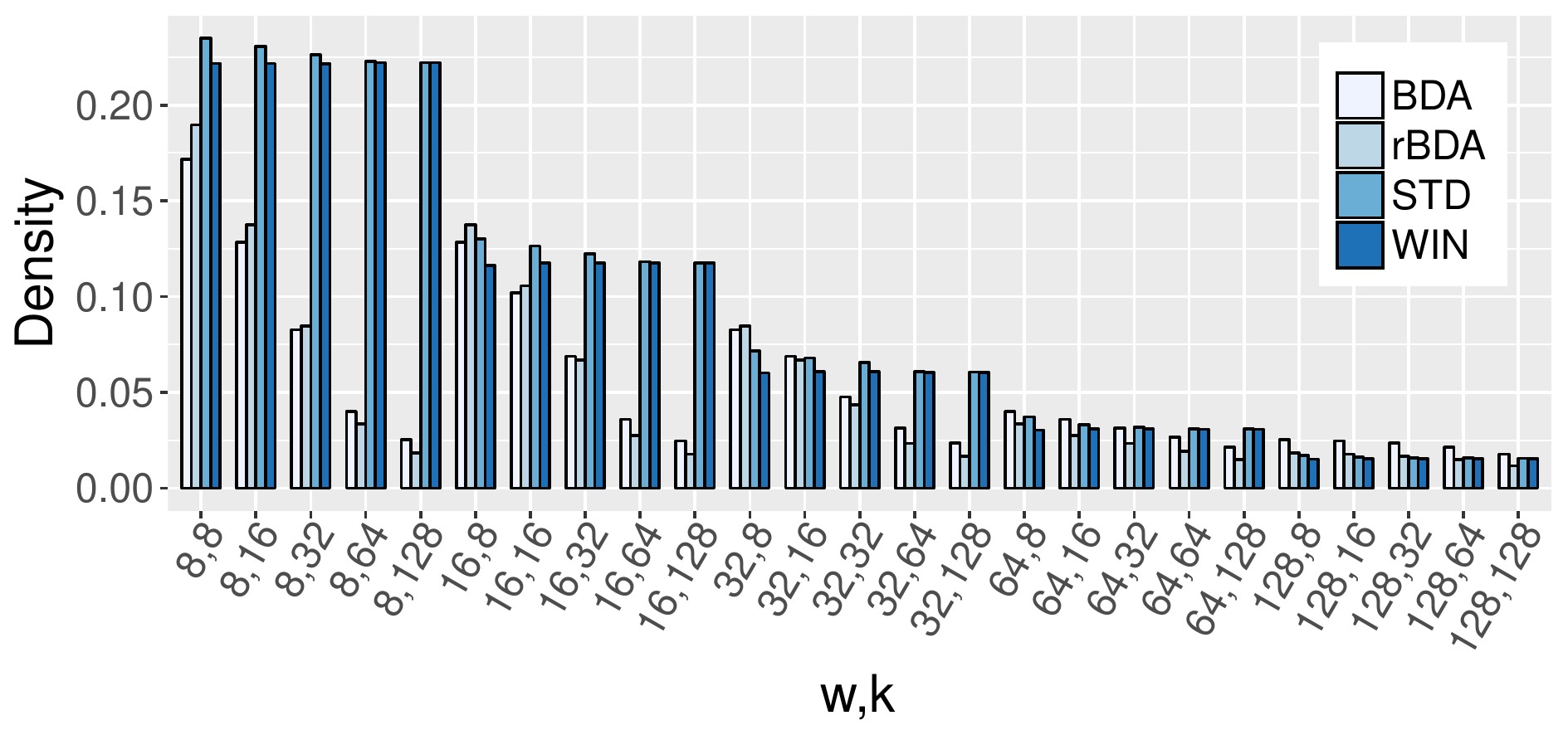}
    \caption{SOURCES}\label{SOURCES}
    \end{subfigure}
    
    \caption{Density vs.~$w,k$ for $\ell=w+k-1$ and the last three datasets of Table~\ref{tab:data}.}\label{fig:experiment1b}
\end{figure}

As can be seen by the results depicted in Figures~\ref{fig:experiment1a} and~\ref{fig:experiment1b}, the density of both BDA and rBDA is either significantly smaller than or competitive to the STD and WIN minimizers density, especially for small $w$. This is useful because a lower density results in smaller indexes and less computation (see Section~\ref{sec:index}), and because small $w$ is of practical interest (see Section~\ref{sec:edit}). For instance, the widely-used long-read aligner \textsf{Minimap2}~\cite{DBLP:journals/bioinformatics/Li18} stores the selected minimizers of a reference genome in a hash table to find exact matches as anchors for seed-and-extend alignment. The parameters $w$ and $k$ are set based on the required sensitivity of the alignment, and thus $w$ and $k$ cannot be too large for high sensitivity. Thus, a lower sampling density reduces the size of the hash table, as well as the computation time, by lowering the average number of selected minimizers to consider when performing an alignment. Furthermore, although the datasets are not uniformly random, rBDA performs better than BDA as $\ell$ grows, as suggested by Lemmas~\ref{lem:randombound} and \ref{lem:reduced}. 

There exists a long line of research on improving the density of minimizers in special regimes (see Section~\ref{sec:related} for details). We stress that most of these algorithms are designed, implemented, or  optimized, \emph{only for the DNA alphabet}. We have tested against two state-of-the-art tools employing such algorithms: \textsf{Miniception}~\cite{DBLP:journals/bioinformatics/ZhengKM20} and \textsf{PASHA}~\cite{DBLP:conf/recomb/EkimBO20}. The former did not give better results than STD or WIN for the tested values of $w$ and $k$; and the latter does not scale beyond $k=16$ or with large alphabets.  We have thus omitted these results.

We next report the average number (AVG) of bd-anchors of order $\ell\in\{4,8,12,16\}$ over all strings of length $n=20$ (see Table~\ref{tab:avg1}) and over all strings of length $n=32$ (see Table~\ref{tab:avg2}), both over a binary alphabet. Notably, the results show that AVG always lies in $[n/\ell,2n/\ell]$ even if not using the reduced version of bd-anchors (see Lemma~\ref{lem:reduced}). As expected by Lemma~\ref{lem:randombound}, the analogous AVG values using a ternary alphabet (not reported here) were always lower than the corresponding ones with a binary alphabet.

\begin{table}[!ht]
\begin{subtable}{.45\linewidth}\centering
\begin{tabular}{|c||c|c|c|c|}\hline
$(n,\ell)$ & $(20,4)$ & $(20,8)$ & $(20,12)$ & $(20,16)$\\\hline\hline
\cline{1-5}
$n/\ell$ & 5 & 2.5 & 1.66 & 1.25 \\\cline{1-5}
\textbf{AVG}       &8.53 & 4.37 & 2.77 & 1.76 \\ \cline{1-5}
$2n/\ell$ &16 & 8 & 5.33 & 4 \\
\cline{1-5}
\end{tabular}\caption{}\label{tab:avg1}
\end{subtable}\hspace{+7mm}
\begin{subtable}{.45\linewidth}\centering
\begin{tabular}{|c||c|c|c|c|}\hline
$(n,\ell)$ & $(32,4)$ & $(32,8)$ & $(32,12)$ & $(32,16)$\\\hline\hline
\cline{1-5}
$n/\ell$ & 8 & 4 & 2,66 & 2 \\\cline{1-5}
\textbf{AVG} &14.16 & 7.67 & 5.26 & 3.85 \\ \cline{1-5}
$2n/\ell$ &16 & 8 & 5.33 & 4 \\
\cline{1-5}
\end{tabular}\caption{}\label{tab:avg2}
\end{subtable}
\caption{Average number of bd-anchors for varying $\ell$ and: (a) $n=20$ and (b) $n=32$.}\label{tab:comparison}
\end{table}

\paragraph{Minimizing the Size of $\G_{\ell}$ is NP-hard.}~The number of bd-anchors depends on the lexicographic total order defined on $\Sigma$. We now prove that finding the total order which minimizes the number of bd-anchors is NP-hard using a reduction from minimum feedback arc set~\cite{DBLP:conf/coco/Karp72}. Let us start be defining this problem. Given a directed graph $G(V,E)$, a \emph{feedback arc set} in $G$ is a subset of $E$ that contains at least one edge from every cycle in $G$. Removing these edges from $G$ breaks all of the cycles, producing a directed acyclic graph. In the \emph{minimum feedback arc set} problem, we are given $G(V,E)$, and we are asked to compute a smallest feedback arc set in $G$. 
The decision version of the problem takes an additional parameter $k$ as input, and it asks whether all cycles can be broken by removing at most $k$ edges from $E$. The decision version is NP-complete~\cite{DBLP:conf/coco/Karp72} and the optimization version is APX-hard~\cite{DBLP:conf/stoc/DinurS02}.

\begin{theorem}
Let $T$ be a string of length $n$ over alphabet $\Sigma$.
Further let $\ell>0$ be an integer.
Computing a total order $\leq$ on $\Sigma$ which minimizes $|\mathcal{A}_\ell(T)|$ is NP-hard.
\end{theorem}
\begin{proof}
Let $G = (\Sigma, A)$ be any instance of the minimum feedback arc set problem. We will construct a string $S \in \Sigma^*$ in polynomial time in the size of $G$ such that finding a total order $\leq$ on $\Sigma$ which minimizes $\mathcal{A}_4(S)$ corresponds to finding a minimum feedback arc set in $G$.

We start with an empty string $S$. For each edge $(a, b) \in A$, we append $(aabab)^{2|A|+1}$ to $S$. Observe that, if $a < b$, all the
$a$'s of $(aabab)^{2|A|+1}$ and none of its $b$'s are order-4 bd-anchors, except possibly the first $a$ and last $b$ depending on the preceding and subsequent letters in $S$. Thus there are $6|A|+2$ to $6|A|+4$ order-$4$ bd-anchors in $(aabab)^{2|A|+1}$. If on the other hand $a > b$, we analogously find that there will be $4|A|+1$ to $4|A|+3$ order-$4$ bd-anchors in $(aabab)^{2|A|+1}$.

Let $d_\leq$ be the number of edges $(a, b) \in A$ such that $a < b$. The total number of order-$4$ bd-anchors in $S$ is $$|\mathcal{A}_4(S)|=|A|\cdot2\cdot (2|A|+1) + d_\leq \cdot (2|A|+1)+ \epsilon,$$ with $\epsilon \in [-|A|, |A|]$. Therefore minimizing the total number of order-$4$ bd-anchors in $S$ is equivalent to finding an order $\leq$ on the set $\Sigma$ of vertices of $G$ which minimizes $d_\leq$.

Note that if we delete all edges $(a, b) \in A$ such that $a < b$, then the residual graph is acyclic. Moreover, for each acyclic graph there exists an order on the vertices such that $a > b$ for all $(a, b) \in A$. Therefore the minimal $d_\leq$ equals the size of the minimum feedback arc set.

We conclude that, since finding the size of the minimum feedback arc set is NP-hard, so is finding a total order $\leq$ on $\Sigma$ which minimizes the total number of order-4 bd-anchors.
\end{proof}

\section{Indexing Using Bidirectional Anchors}\label{sec:index}

Before presenting our index, let us start with a basic definition that is central to our querying process.

\begin{definition}[$(\alpha,\beta)$-hit]
Given an order-$\ell$ bd-anchor $j_Q\in \G_{\ell}(Q)$, for some integer $\ell>0$, of a query string $Q$, two integers $\alpha> 0,\beta> 0$, with $\alpha+\beta\geq \ell+1$, and an order-$\ell$ bd-anchor $j_T \in \G_{\ell}(T)$ of a target string $T$, the ordered pair $(j_Q,j_T)$ is called an \emph{$(\alpha,\beta)$-hit} if and only if $T[j_T-\alpha+1\dd j_T]=Q[j_Q-\alpha+1 \dd j_Q]$ and $T[j_T\dd j_T+\beta-1]=Q[j_Q\dd j_Q+\beta-1]$.
\end{definition}

Intuitively, the parameters $\alpha$ and $\beta$ let us choose a fragment of $Q$ that is anchored at $j_Q$.

\begin{example}
Let $T=\texttt{aabaaabcbda}$, $Q=\texttt{aacabaaaae}$, and $\ell=5$. Consider that we would like to find the common fragment $Q[4\dd 8]=T[2\dd 6]=\texttt{abaaa}$. 
We know that the bd-anchor of order $5$ corresponding to $Q[4\dd 8]$ is $6\in\G_5(Q)$, and thus to find it we set $\alpha=3$ and $\beta=3$. The ordered pair $(6,4)$ is a $(3,3)$-hit because for $4\in\G_5(T)$, we have: $T[4-3+1\dd 4]=Q[6-3+1 \dd 6]=\texttt{aba}$ and $T[4\dd 4+3-1]=Q[6 \dd 6+3-1]=\texttt{aaa}$.
\label{example4}
\end{example}

We would like to construct a data structure over $T$, which is based on $\G_{\ell}(T)$, such that, when we are given an order-$\ell$ bd-anchor $j_Q$ over $Q$ as an on-line query, together with parameters $\alpha$ and $\beta$, we can report all $(\alpha,\beta)$-hits efficiently. To this end, we present an efficient data structure, denoted by $\mathcal{I}_{\ell}(T)$, which is constructed on top of $T$, and answers $(\alpha,\beta)$-hit queries in near-optimal time. 
We prove the following result.

\begin{theorem}\label{the:main}
Given a string $T$ of length $n$ and an integer $\ell>0$, the $\mathcal{I}_{\ell}(T)$ index can be constructed in $\cO(n +|\G_{\ell}(T)|\sqrt{\log(|\G_{\ell}(T)|)})$ time. For any constant $\epsilon>0$, $\mathcal{I}_{\ell}(T)$:
\begin{itemize}
    \item  occupies $\cO(|\G_{\ell}(T)|)$ extra space and reports all $k$ $(\alpha,\beta)$-hits in $\cO(\alpha+\beta+(k+1)\log^{\epsilon}(|\G_{\ell}(T)|))$ time; or
    \item  occupies $\cO(|\G_{\ell}(T)|\log^{\epsilon}(|\G_{\ell}(T)|))$ extra space and reports all $k$ $(\alpha,\beta)$-hits in $\cO(\alpha+\beta+\log\log(|\G_{\ell}(T)|)+k)$ time.
\end{itemize}
\end{theorem}

Let us denote by $\overleftarrow{X}=X[|X|]\ldots X[1]$ the \emph{reversal} of string $X$. We now describe our data structure.

\paragraph{Construction of $\mathcal{I}_{\ell}(T)$.}~Given $\G_{\ell}(T)$, we construct two sets $\R^{L}_{\ell}(T)$ and $\R^{R}_{\ell}(T)$ of strings; conceptually, the reversed suffixes going \emph{left} from $j$ to $1$, and the suffixes going \emph{right} from $j$ to $n$, for all $j$ in $\G_{\ell}(T)$. In particular, for the bd-anchor $j$, we construct two strings: $\overleftarrow{T[1\dd j]}\in \R^{L}_{\ell}(T)$ and $T[j\dd n]\in \R^{R}_{\ell}(T)$. Note that, $|\R^{L}_{\ell}(T)|= |\R^{R}_{\ell}(T)|=|\G_{\ell}(T)|$, since for every bd-anchor in $\G_{\ell}(T)$ we have a distinct string in $\R^{L}_{\ell}(T)$ and in $\R^{R}_{\ell}(T)$. 

We construct two \emph{compacted tries} $\T^{L}_{\ell}(T)$ and $\T^{R}_{\ell}(T)$ over $\R^{L}_{\ell}(T)$ and $\R^{R}_{\ell}(T)$, respectively, to index all strings. Every string is concatenated with some special letter $\$$ not occurring in $T$, which is lexicographically minimal, to make $\R^{L}_{\ell}(T)$ and $\R^{R}_{\ell}(T)$ prefix-free (this is standard for conceptual convenience).
The leaf nodes of the compacted tries are \emph{labeled} with the corresponding $j$: there is a one-to-one correspondence between a leaf node and a bd-anchor $j$. In $\cO(|\G_{\ell}(T)|)$ time, we also enhance the nodes of the tries with a perfect static dictionary~\cite{Fredman:1984:SST:828.1884} to ensure constant-time retrieval of edges by the first letter of their label. Let $\A^{L}_{\ell}(T)$  
denote the list of the leaf labels of $\T^{L}_{\ell}(T)$ 
as they are visited using a depth-first search traversal. $\A^{L}_{\ell}(T)$  
corresponds to the (labels of the) lexicographically sorted list of $\R^{L}_{\ell}(T)$ 
in increasing order. For each node $u$ 
in $\T^{L}_{\ell}(T)$,  
we also store the corresponding interval $[x_u,y_u]$ over $\A^{L}_{\ell}(T)$.  
Analogously for $R$,  $\A^{R}_{\ell}(T)$  denotes the list of the leaf labels of $\T^{R}_{\ell}(T)$ as they are visited using a depth-first search traversal and corresponds to the (labels of the) lexicographically sorted list of $\R^{R}_{\ell}(T)$ in increasing order. For each node $v$ in $\T^{R}_{\ell}(T)$, we also store the corresponding interval $[x_v,y_v]$ over $\A^{R}_{\ell}(T)$. 

The total size occupied by the tries is $\Theta(|\G_{\ell}(T)|)$ because they are compacted: we label the edges with intervals over $[1,n]$ from $T$.

We also construct a 2D range reporting data structure over the following points in set $\RANGE_\ell(T)$: \[(x,y)\in \RANGE_\ell(T) \iff \A^{L}_{\ell}(T)[x]=\A^{R}_{\ell}(T)[y].\] Note that $|\RANGE_\ell(T)|=|\G_{\ell}(T)|$ because the set of leaf labels stored in both tries is precisely the set $\G_{\ell}(T)$. Let us remark that the idea of employing 2D range reporting for bidirectional pattern searches has been introduced by Amir et al.~\cite{DBLP:journals/jal/AmirKLLLR00} for text indexing and dictionary matching with one error; see also~\cite{DBLP:conf/latin/MakinenN06}.

This completes the construction of $\mathcal{I}_{\ell}(T)$. We next explain how we can query $\mathcal{I}_{\ell}(T)$.

\paragraph{Querying.}~Given a bd-anchor $j_Q$ over a string $Q$ as an on-line query and parameters $\alpha,\beta>0$, we spell $\overleftarrow{Q[j_Q-\alpha+1\dd j_Q]}$ in $\T^{L}_{\ell}(T)$ and $Q[j_Q\dd j_Q+\beta-1]$ in $\T^{R}_{\ell}(T)$ starting from the root nodes. If any of the two strings is not spelled fully, we return no $(\alpha,\beta)$-hits.
If both strings are fully spelled, we arrive at node $u$ in $\T^{L}_{\ell}(T)$ (resp.~$v$ in $\T^{R}_{\ell}(T)$), which corresponds to an interval over $\A^{L}_{\ell}(T)$ stored in $u$ (resp.~$\A^{R}_{\ell}(T)$ in $v$). We obtain the two intervals $[x_u,y_u]$ and $[x_v,y_v]$ forming a rectangle and ask the corresponding 2D range reporting query. It can be readily verified that this query returns all $(\alpha,\beta)$-hits.

\begin{example}
Let $T=\texttt{aabaaabcbda}$ and $\G_5(T)=\{4,5,6,11\}$. 
We have the following strings in $\R^L(T)$: $\overleftarrow{T[1\dd 4]}=\texttt{abaa}$; $\overleftarrow{T[1\dd 5]}=\texttt{aabaa}$; $\overleftarrow{T[1\dd 6]}=\texttt{aaabaa}$; and $\overleftarrow{T[1\dd 11]}=\texttt{adbcbaaabaa}$.
We have the following strings in $\R^R(T)$: $T[4 \dd 11]=\texttt{aaabcbda}$; $T[5 \dd 11]=\texttt{aabcbda}$; $T[6 \dd 11]=\texttt{abcbda}$; $T[11 \dd 11]=\texttt{a}$.
Inspect Figure~\ref{fig:matching}.
\end{example}

\begin{figure}[t]
     \centering
     \begin{subfigure}[b]{0.45\textwidth}
         \centering
         \includegraphics[width=6cm]{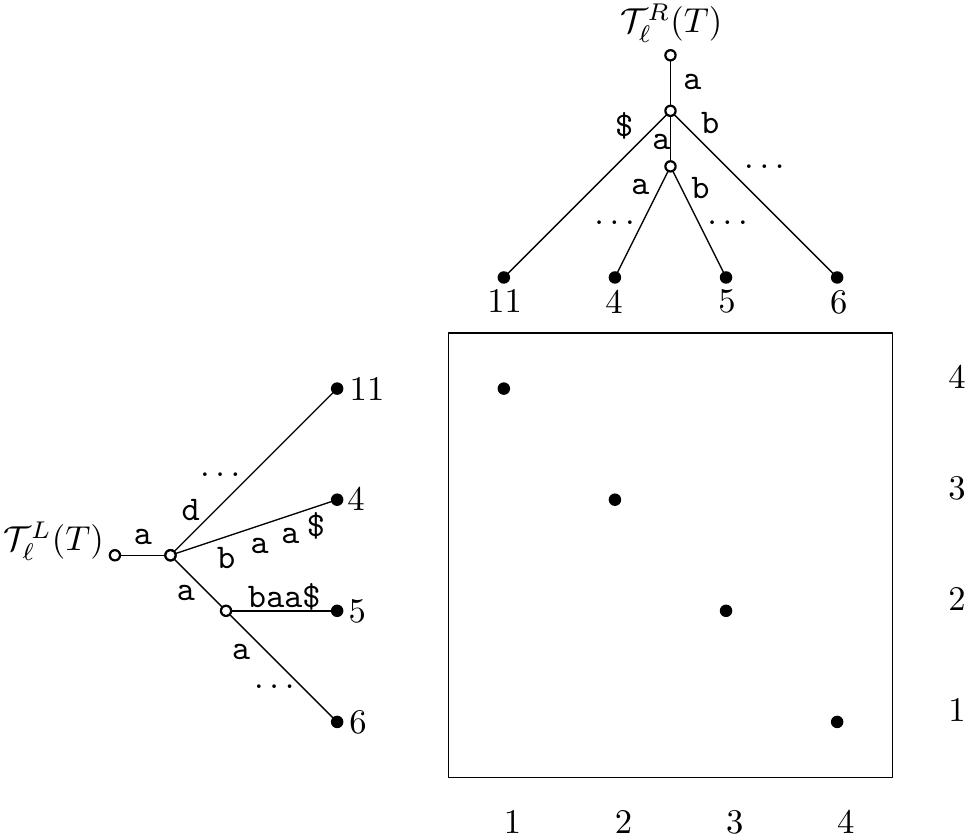}
         \caption{The $\mathcal{I}_{\ell}(T)$ index}
         \label{fig:ds}
     \end{subfigure}
     \hfill
     \begin{subfigure}[b]{0.45\textwidth}
         \centering
         \includegraphics[width=6cm]{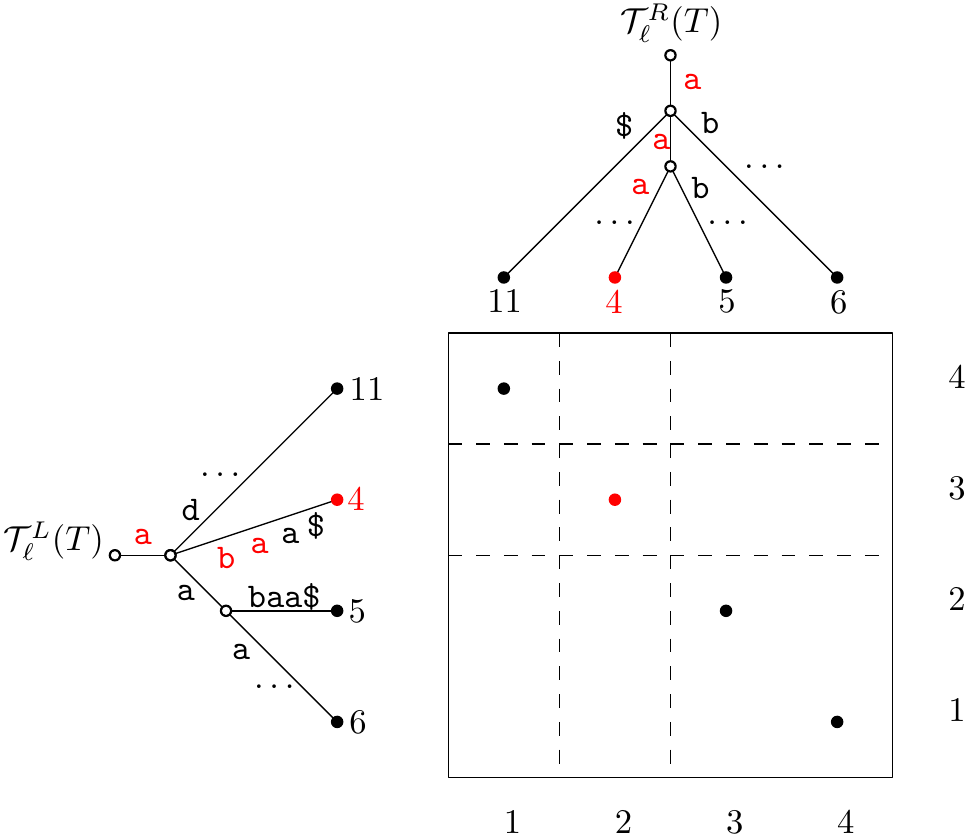}
         \caption{Querying \texttt{abaaa}}
         \label{fig:q}
     \end{subfigure}
         \caption{Let $T=\texttt{aabaaabcbda}$ and $\ell=5$. Further let $Q=\texttt{aacabaaaae}$, the bd-anchor $6\in\G_5(Q)$ of order $5$ corresponding to $Q[4\dd 8]$, $\alpha=3$ and $\beta=3$. The figure illustrates the $\mathcal{I}_{\ell}(T)$ index and how we find that $Q[4\dd 8]=T[2\dd 5]=\texttt{abaaa}$: the fragment $T[2\dd 5]$ is anchored at position $4$.}
       \label{fig:matching}
\end{figure}

\begin{proof}[Proof of Theorem~\ref{the:main}]
We use the $\cO(n)$-time algorithm underlying Theorem~\ref{the:con} to construct $\G_{\ell}(T)$.
We use the $\cO(n)$-time algorithm from~\cite{DBLP:journals/iandc/BartonK0PR20,sorting} to construct the compacted tries from $\G_{\ell}(T)$.
We extract the $|\G_{\ell}(T)|$ points $(x,y)\in \RANGE_\ell(T)$ using the compacted tries in $\cO(|\G_{\ell}(T)|)$ time.
For the first trade-off  of the statement, we use the $\cO(|\G_{\ell}(T)|\sqrt{\log(|\G_{\ell}(T)|)})$-time algorithm from~\cite{DBLP:conf/soda/BelazzouguiP16} to construct the 2D range reporting data structure over $\RANGE_\ell(T)$ from~\cite{DBLP:conf/compgeom/ChanLP11}. For the second trade-off, we use the $\cO(|\G_{\ell}(T)|\sqrt{\log(|\G_{\ell}(T)|)})$-time algorithm from~\cite{gao_et_al:LIPIcs:2020:12920} to construct the 2D range reporting data structure over $\RANGE_\ell(T)$ from the same paper. \end{proof}

We obtain the following corollary for the fundamental problem of \emph{text indexing}~\cite{DBLP:conf/focs/Weiner73,DBLP:journals/siamcomp/ManberM93,DBLP:conf/focs/Farach97,DBLP:journals/jacm/KarkkainenSB06,DBLP:journals/jacm/FerraginaM05,DBLP:journals/siamcomp/GrossiV05,DBLP:journals/siamcomp/HonSS09,DBLP:conf/stoc/Belazzougui14,DBLP:journals/algorithmica/0001KL15,DBLP:conf/soda/MunroNN17,DBLP:conf/stoc/KempaK19,DBLP:journals/jacm/GagieNP20}.

\begin{corollary}\label{the:indexing}
Given $\mathcal{I}_\ell(T)$ constructed for some integer $\ell>0$ and some constant $\epsilon>0$ over string $T$, we can report all $k$ occurrences of any pattern $Q$, $|Q|\geq \ell$, in $T$ in time:
\begin{itemize}
    \item $\cO(|Q|+(k+1)\log^{\epsilon}(|\G_{\ell}(T)|))$ when $\mathcal{I}_\ell(T)$ occupies $\cO(|\G_{\ell}(T)|)$ extra space; or
    \item $\cO(|Q|+\log\log(|\G_{\ell}(T)|)+k)$ when $\mathcal{I}_\ell(T)$ occupies $\cO(|\G_{\ell}(T)|\log^{\epsilon}(|\G_{\ell}(T)|))$ extra space.
\end{itemize}
\end{corollary}
\begin{proof}
Every occurrence of $Q$ in $T$ is prefixed by string $P=Q[1\dd \ell]$.
We first compute the bd-anchor of $P$ in $\cO(\ell)$ time using Booth's algorithm. Let this bd-anchor be $j$. We set $\alpha=j$ and $\beta=|Q|-j+1$. The result follows by applying Theorem~\ref{the:main}.
\end{proof}

\paragraph{Querying Multiple Fragments.}
In the case of approximate pattern matching, we may want to query multiple length-$\ell$ fragments of a string $Q$ given as an on-line query, and not only its length-$\ell$ prefix. We show that such an operation can be done efficiently using the bd-anchors of $Q$ and the $\mathcal{I}_\ell(T)$ index.

\begin{corollary}
Given $\mathcal{I}_\ell(T)$ constructed for some $\ell>0$ and some constant $\epsilon>0$ over $T$, for any sequence (not necessarily consecutive) of $d>0$ length-$\ell$ fragments of a pattern $Q$, $|Q|\geq \ell$, corresponding to the same order-$\ell$ bd-anchor of $Q$, we can report all $k_d$ occurrences of all $d$ fragments in $T$ in time:
\begin{itemize}
    \item $\cO(\ell+(d+k_d)\log^{\epsilon}(|\G_{\ell}(T)|))$ when $\mathcal{I}_\ell(T)$ occupies $\cO(|\G_{\ell}(T)|)$ space; or
    \item $\cO(\ell+d\log\log(|\G_{\ell}(T)|)+k_d)$ when $\mathcal{I}_\ell(T)$ occupies $\cO(|\G_{\ell}(T)|\log^{\epsilon}(|\G_{\ell}(T)|))$ space.
\end{itemize}
\end{corollary}

\begin{proof}
Let the order-$\ell$ bd-anchor over $Q$ be $j_Q$ and the corresponding parameters be $(\alpha_1,\beta_1),\cdots,(\alpha_d,\beta_d)$, with $\alpha_i+\beta_i=\ell+1$. Observe that $\alpha_i>\alpha_{i+1}$ and $\beta_i<\beta_{i+1}$. Starting from $j_Q$, the string $\overleftarrow{Q[j_Q-\alpha_i+1\dd j_Q]}$ we spell for fragment $i$ is the prefix of $\overleftarrow{Q[j_Q-\alpha_{i-1}+1\dd j_Q]}$ for fragment $i-1$. The analogous property holds for the other direction: the string $Q[j_Q\dd j_Q+\beta_i]$ we spell for fragment $i$ is the prefix of $Q[j_Q\dd j_Q+\beta_{i+1}]$ for fragment $i+1$. Thus it takes only $\cO(\ell)$ time to construct all $d$ rectangles. Finally, we ask the $d$ corresponding 2D range reporting queries to obtain all $k_d$ occurrences in the claimed time complexities.
\end{proof}

\paragraph{Index Evaluation.}~Consider a hash table with the following (key, value) pairs: the \emph{key} is the hash value $h(S)$ of a length-$k$ string $S$; and the \emph{value} (satellite data) is a list of occurrences of $S$ in $T$. It should be clear that such a hash table indexing the minimizers of $T$ does not perform well for on-line pattern searches of \emph{arbitrary length} because it would need to verify the remaining prefix and suffix of the pattern using letter comparisons \emph{for all occurrences} of a minimizer in $T$. We thus opted for comparing our index to the one of~\cite{DBLP:journals/spe/GrabowskiR17}, which addresses this specific problem by sampling the suffix  array~\cite{DBLP:journals/siamcomp/ManberM93} with minimizers to reduce the number of letter comparisons during verification.

To ensure a fair comparison, we have implemented the basic index from~\cite{DBLP:journals/spe/GrabowskiR17}; we denote it by \textsf{GR Index}. We used Karp-Rabin~\cite{DBLP:journals/ibmrd/KarpR87} fingerprints for computing the minimizers of $T$. We also used the array-based version of the suffix tree that consists of the suffix array (SA) and the longest common prefix (LCP) array~\cite{DBLP:journals/siamcomp/ManberM93}; SA was constructed using SDSL~\cite{DBLP:conf/wea/GogBMP14} and the LCP array using the Kasai et al.~\cite{DBLP:conf/cpm/KasaiLAAP01} algorithm. 
 
We sampled the SA using the minimizers. 
Given a pattern $Q$, we searched $Q[j\dd |Q|]$ starting with the minimizer $Q[j\dd j+k-1]$ using the Manber and Myers~\cite{DBLP:journals/siamcomp/ManberM93} algorithm on the sampled SA. For verifying the remaining prefix $Q[1\dd j-1]$ of $Q$, we used letter comparisons, as described in~\cite{DBLP:journals/spe/GrabowskiR17}.
The space complexity of this implementation is $\cO(n)$ and the extra space for the index is $\cO(|\M_{w,k}(T)|)$. The query time is not bounded. We have implemented two versions of our index. We used Booth's algorithm for computing the bd-anchors of $T$. We used SDSL for SA construction and the Kasai et al.~algorithm for LCP array construction. We sampled the SA using the bd-anchors thus constructing $\A^{L}_{\ell}(T)$ and $\A^{R}_{\ell}(T)$. Then, the two versions of  our index are:
\begin{enumerate}
    \item \textbf{\textsf{BDA Index v1}}: Let $j$ be the bd-anchor of $Q[1\dd \ell]$.
    For $\overleftarrow{Q[1\dd j]}$ (resp.~$Q[j\dd |Q|]$) we used the Manber and Myers algorithm for searching over $\A^{L}_{\ell}(T)$ (resp.~$\A^{R}_{\ell}(T)$). We used range trees~\cite{DBLP:books/lib/BergCKO08} implemented in CGAL~\cite{cgal:eb-21a} for 2D range reporting as per the described querying process. The space complexity of this implementation is $\cO(n+|\G_{\ell}(T)|\log(|\G_{\ell}(T)|))$ and the extra space for the index is $\cO(|\G_{\ell}(T)|\log(|\G_{\ell}(T)|))$.
    The query time is $\cO(|Q| + \log^2(|\G_{\ell}(T)|) + k)$, where $k$ is the total number of occurrences of $Q$ in $T$. 
    \item \textbf{\textsf{BDA Index v2}}: Let $j$ be the  bd-anchor of $Q[1\dd \ell]$. If $|Q|-j+1\geq j$ (resp.~$|Q|-j+1< j$), we search for $Q[j\dd |Q|]$ (resp.~$\overleftarrow{Q[1\dd j]}$) using the Manber and Myers algorithm on $\A^{R}_{\ell}(T)$ (resp.~$\A^{L}_{\ell}(T)$). For verifying the remaining part of the pattern we used letter comparisons. The space complexity of this implementation is $\cO(n)$ and the extra space for the index is $\cO(|\G_{\ell}(T)|)$. The query time is not bounded.
\end{enumerate}

For each of the five real datasets of Table~\ref{tab:data} and each query string length $\ell$, we randomly extracted 500,000 substrings from the text and treated each substring as a query, following~\cite{DBLP:journals/spe/GrabowskiR17}. 
We plot the average query time in Figure~\ref{fig:experiment2-time}. As can be seen, \textsf{BDA Index v2} consistently outperforms \textsf{GR Index} across all datasets and all $\ell$ values. The better performance of \textsf{BDA Index v2} is due to two theoretical reasons. First, the verification strategy exploits the fact that the index is \emph{bidirectional} to apply the Manber and Myers algorithm to the largest part of the pattern, which results in fewer letter comparisons. Second, bd-anchors generally have smaller density compared to minimizers; see Figure~\ref{fig:experiment2-density}. We also plot the peak memory usage in Figure~\ref{fig:experiment2-mem}. As can be seen, \textsf{BDA Index v2} requires a similar amount of memory to \textsf{GR Index}. 

\begin{figure}[t]
    \begin{subfigure}[b]{0.45\textwidth}
    \includegraphics[width=0.9\linewidth]{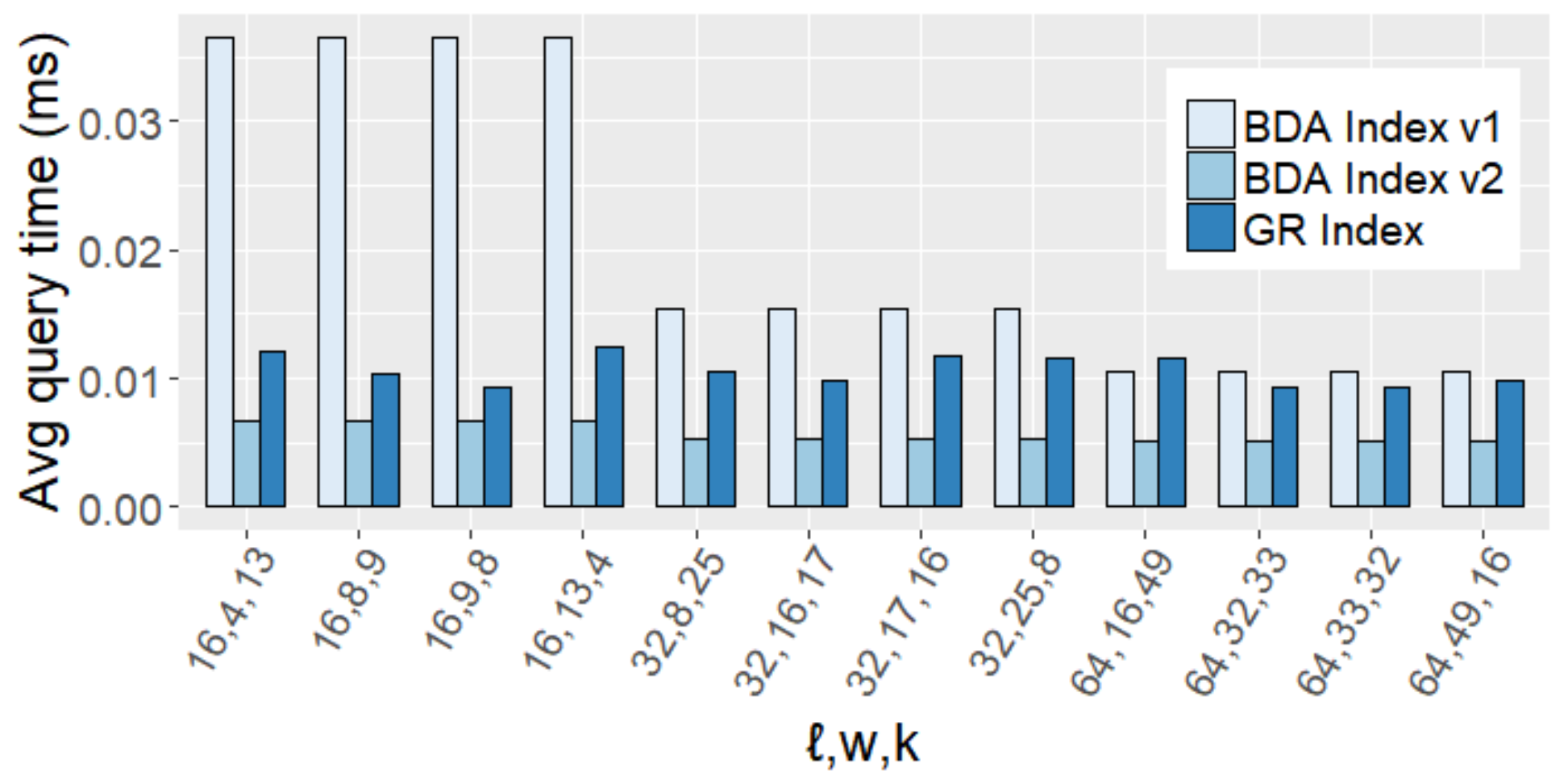}
    \caption{DNA}
    \end{subfigure}\hspace{+1mm}
    \begin{subfigure}[b]{0.45\textwidth}
    \includegraphics[width=0.9\linewidth]{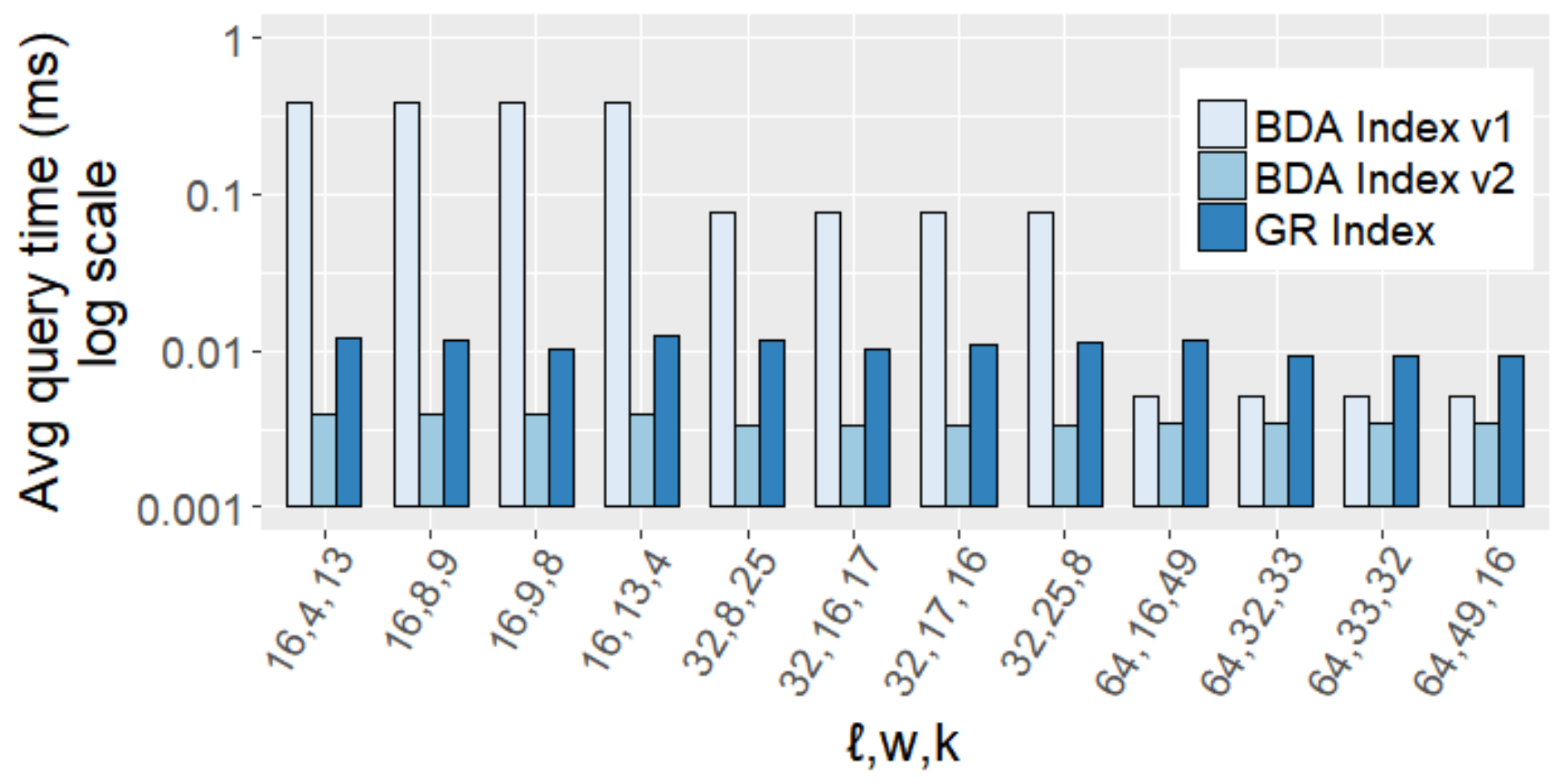}
    \caption{\DBLP}
    \end{subfigure}
    \begin{subfigure}[b]{0.45\textwidth}
    \includegraphics[width=0.9\linewidth]{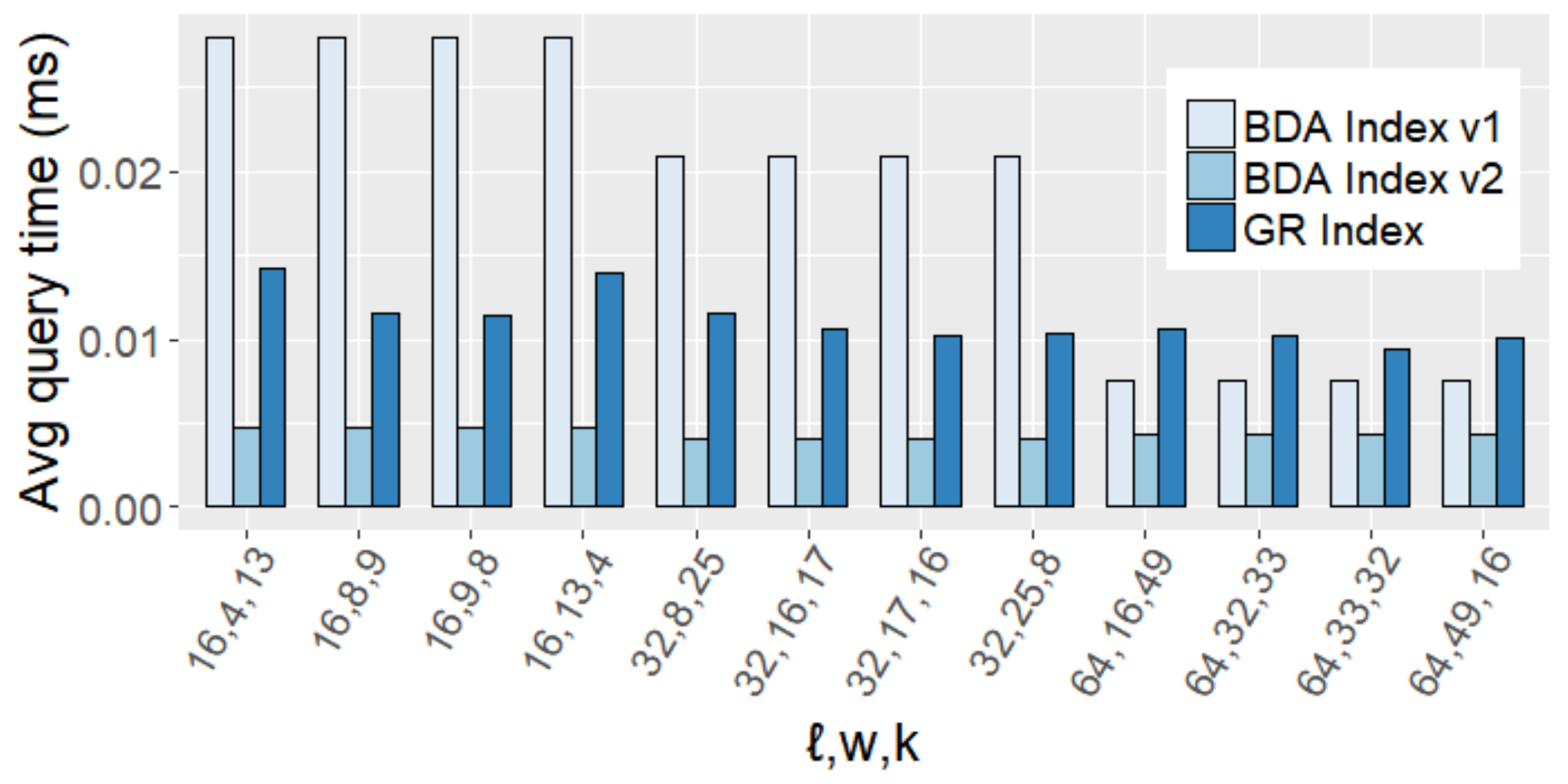}
    \caption{ENGLISH}
    \end{subfigure}
    \begin{subfigure}[b]{0.45\textwidth}
    \includegraphics[width=0.9\linewidth]{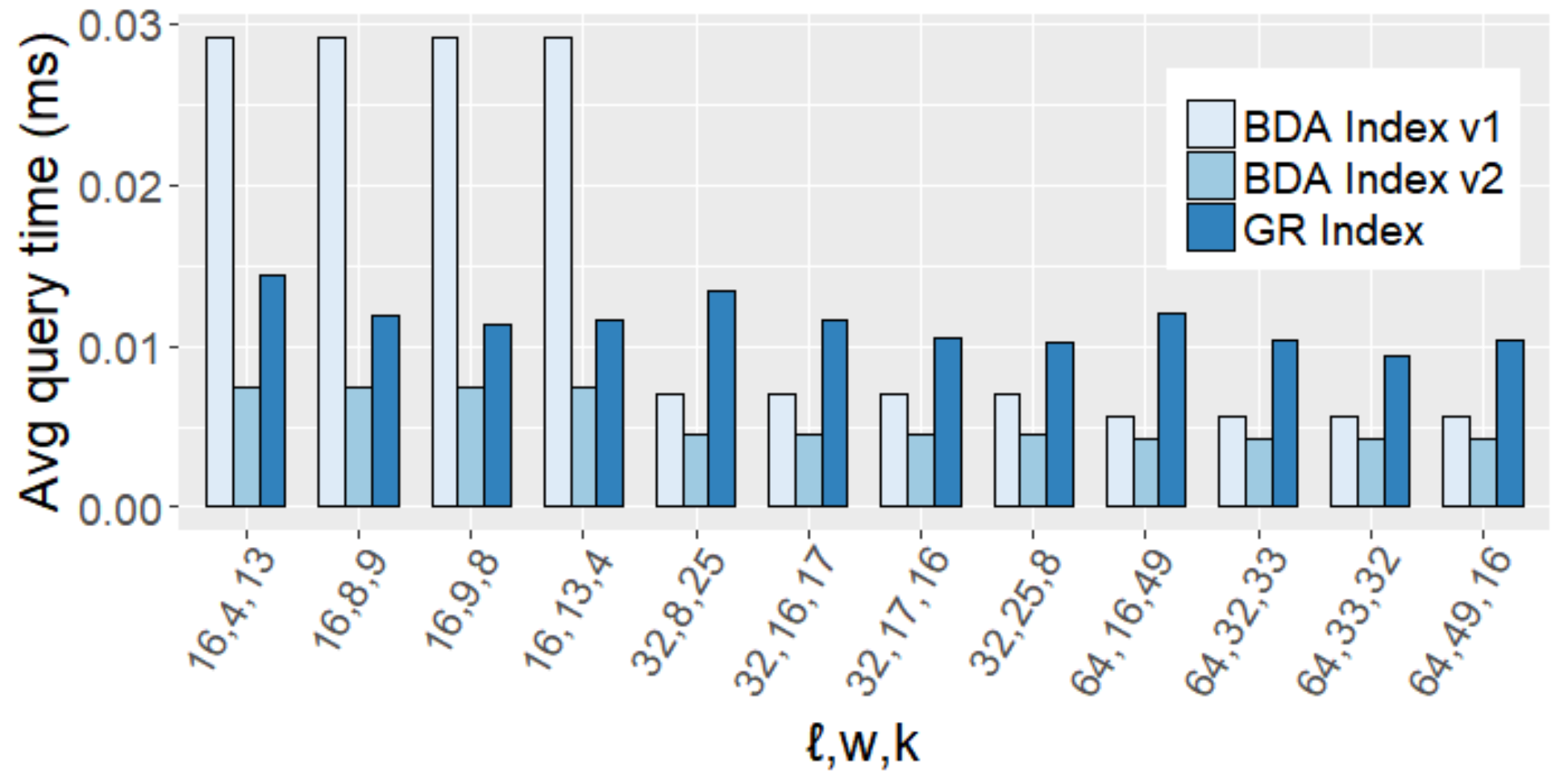}
    \caption{PROTEINS}
    \end{subfigure}
    \centering
    \begin{subfigure}[b]{0.45\textwidth}
    \includegraphics[width=0.9\linewidth]{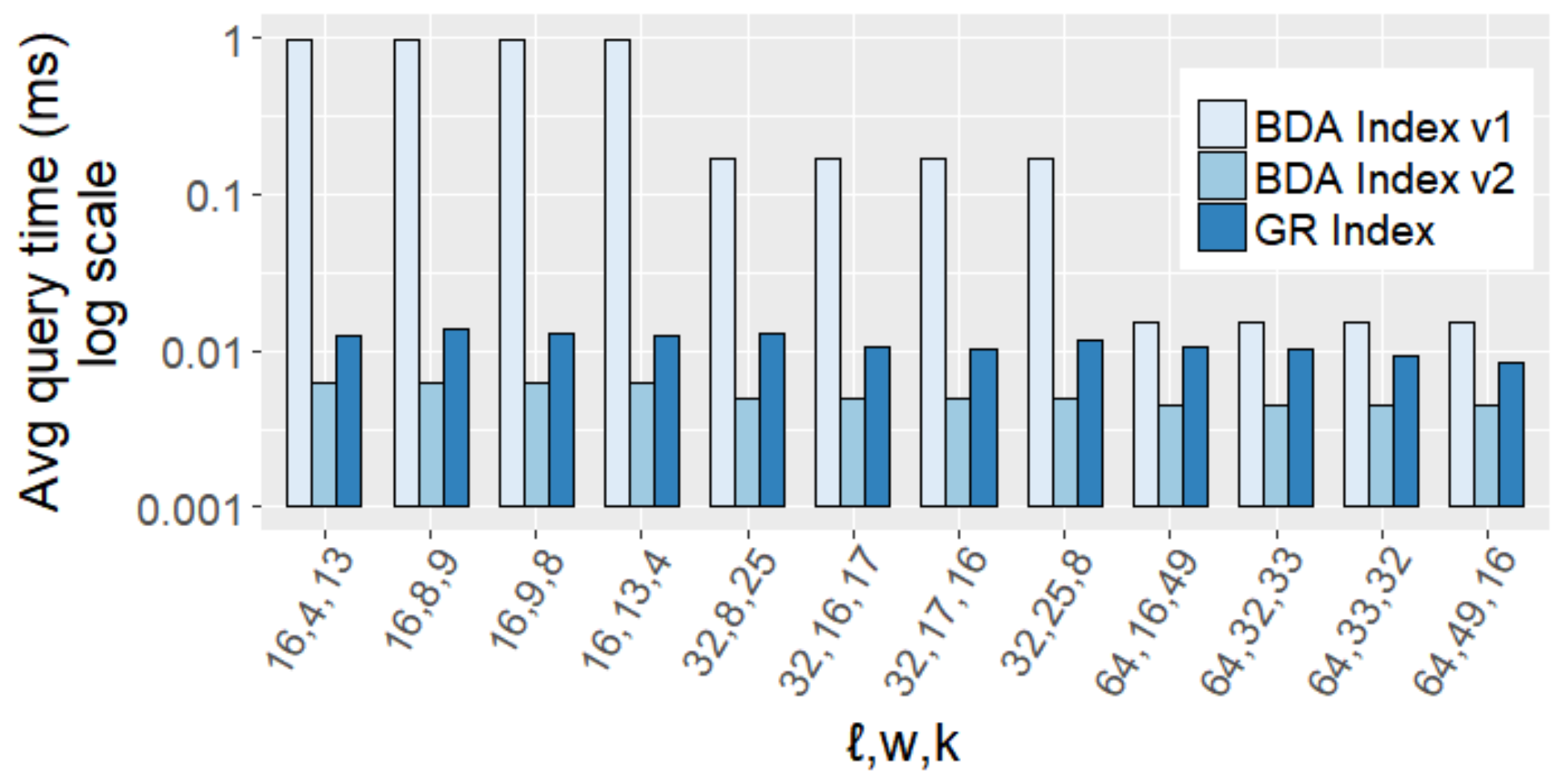}
    \caption{SOURCES}
    \end{subfigure}
    \caption{Average query time (ms) vs.~$w,k$ for $\ell=w+k-1$ and the datasets of Table~\ref{tab:data}.}\label{fig:experiment2-time}
    \end{figure}

\textsf{BDA Index v1} was slower than \textsf{GR Index} for small $\ell$ but faster for large $\ell$ in three out of five datasets used and had by far the highest memory usage. Let us stress that the inefficiency of \textsf{BDA Index v1} is not due to inefficiency in the query time or space of our algorithm. It is merely because  the range tree implementation of CGAL, which is a standard off-the-shelf library, is unfortunately inefficient in terms of both query time and memory usage; see also~\cite{alenex19,fisikopoulos}. 

\begin{figure}[t]
     \begin{subfigure}[b]{0.453\textwidth}
     \includegraphics[width=1\linewidth]{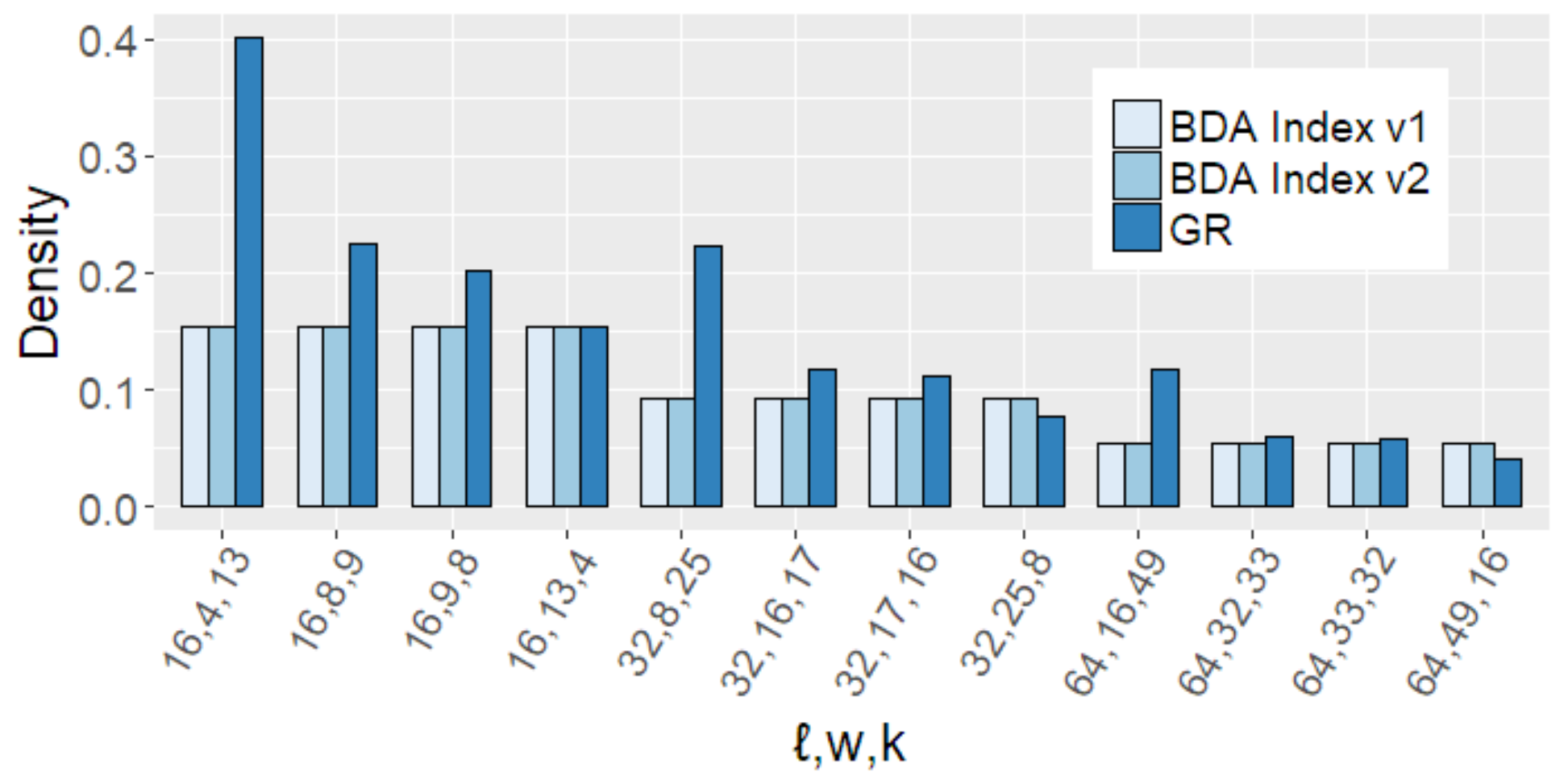}
     \caption{DNA}
     \end{subfigure}\hspace{+2mm}
     \begin{subfigure}[b]{0.453\textwidth}
     \includegraphics[width=1\linewidth]{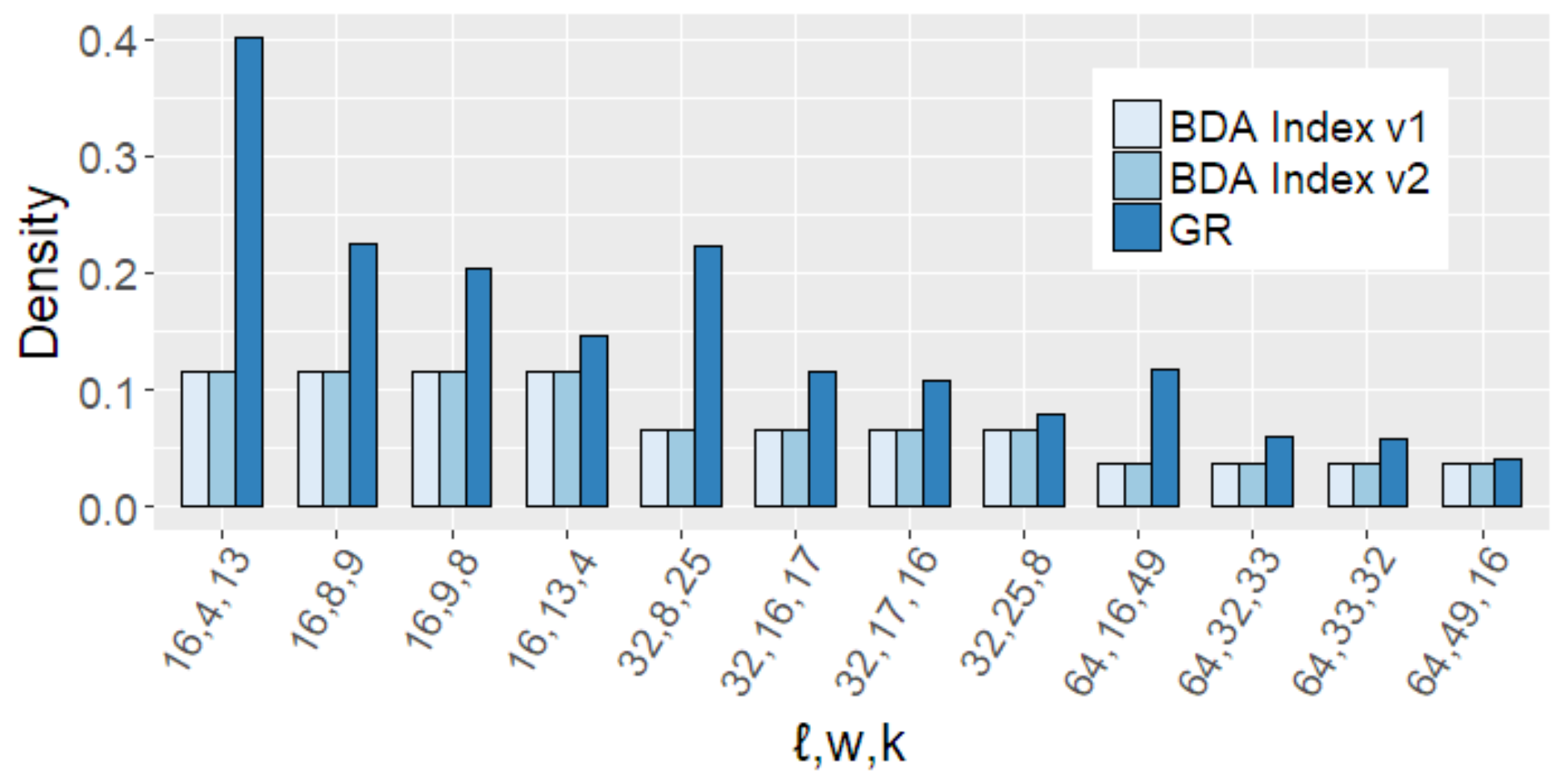}
     \caption{XML}
     \end{subfigure}\\
     \begin{subfigure}[b]{0.453\textwidth}
     \includegraphics[width=1\linewidth]{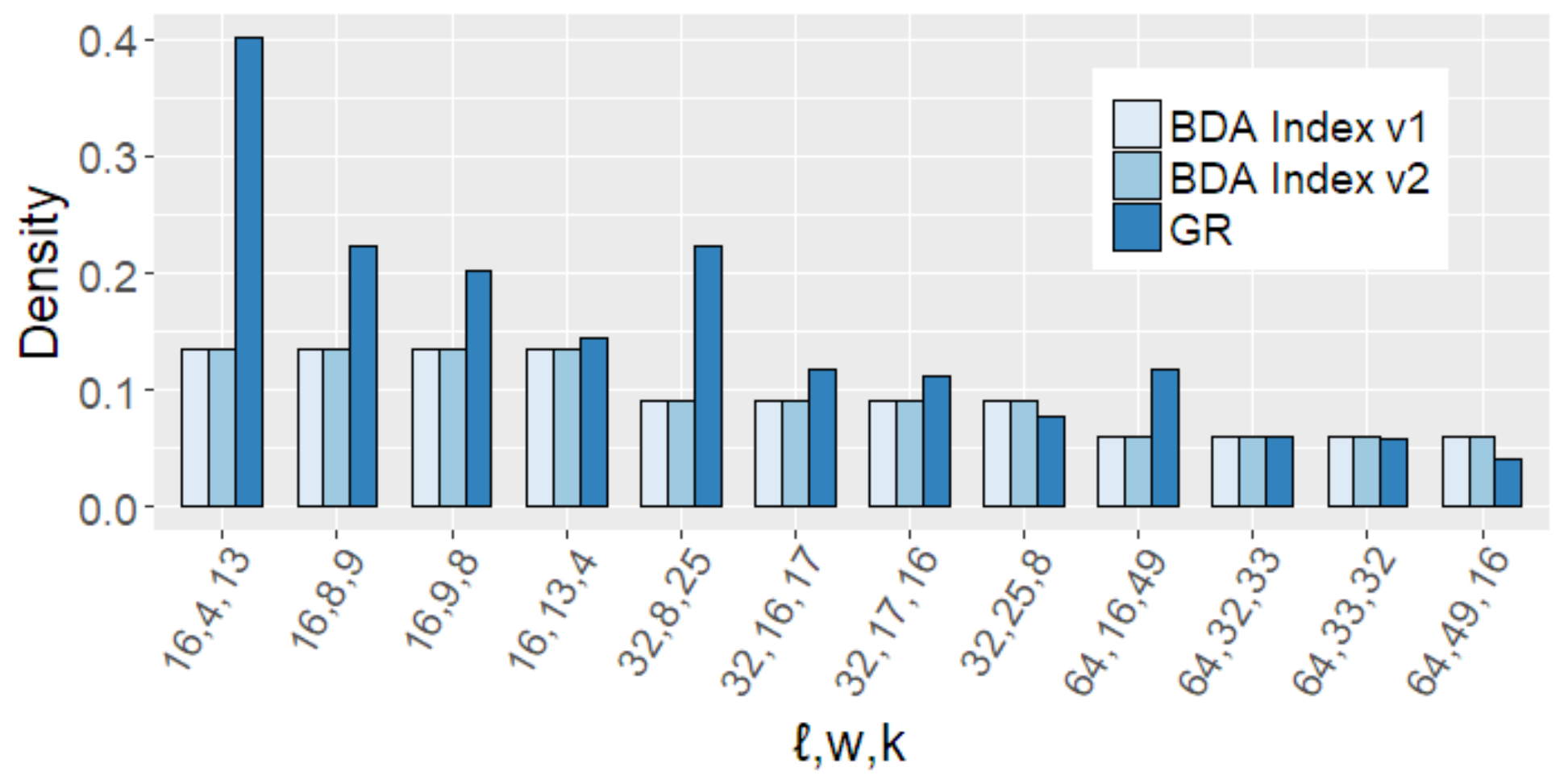}
     \caption{ENGLISH}
     \end{subfigure}\hspace{+2mm}
     \begin{subfigure}[b]{0.453\textwidth}
     \includegraphics[width=1\linewidth]{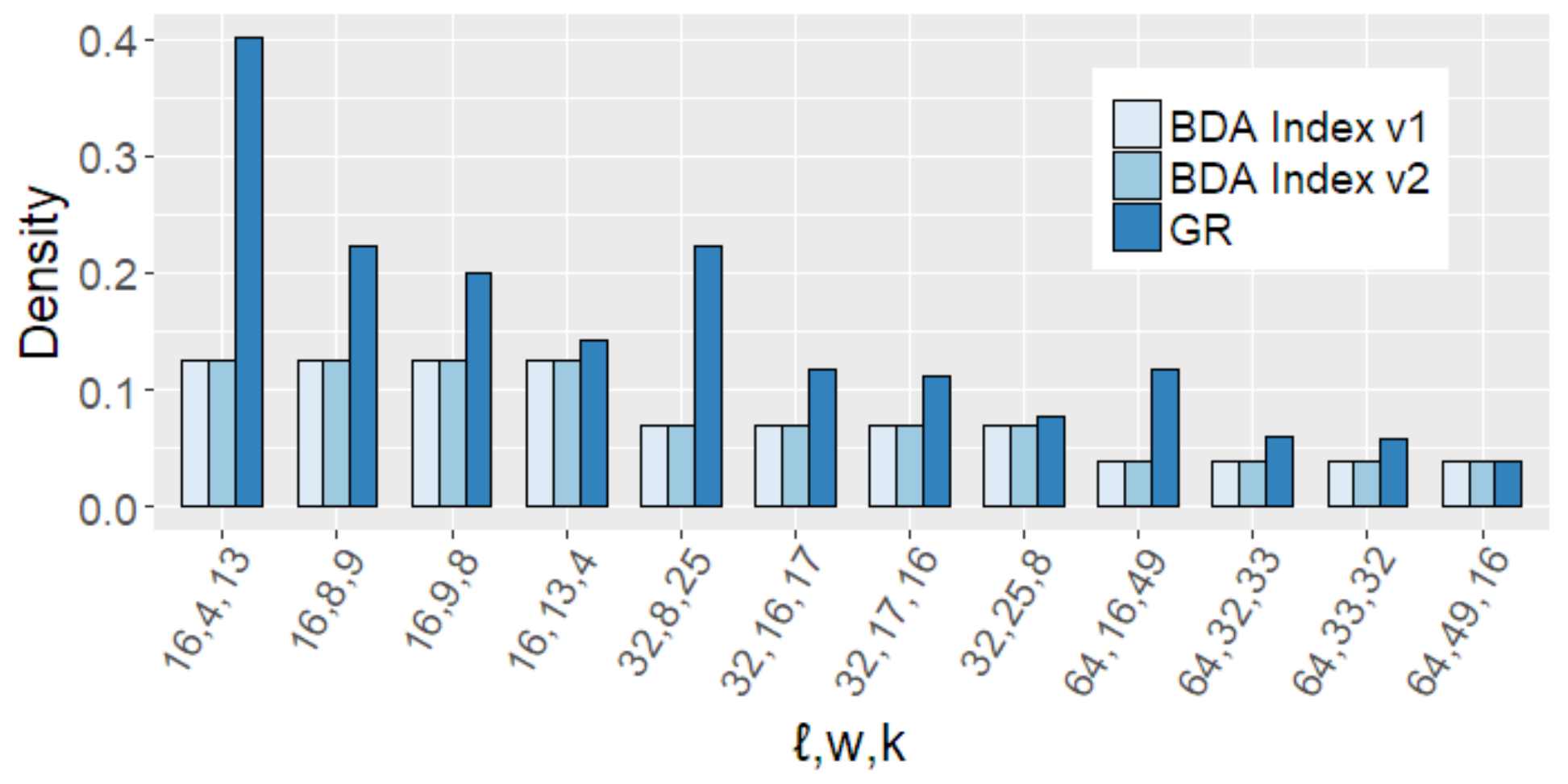}
     \caption{PROTEINS}
     \end{subfigure}\\\centering
     \begin{subfigure}[b]{0.453\textwidth}
     \includegraphics[width=1\linewidth]{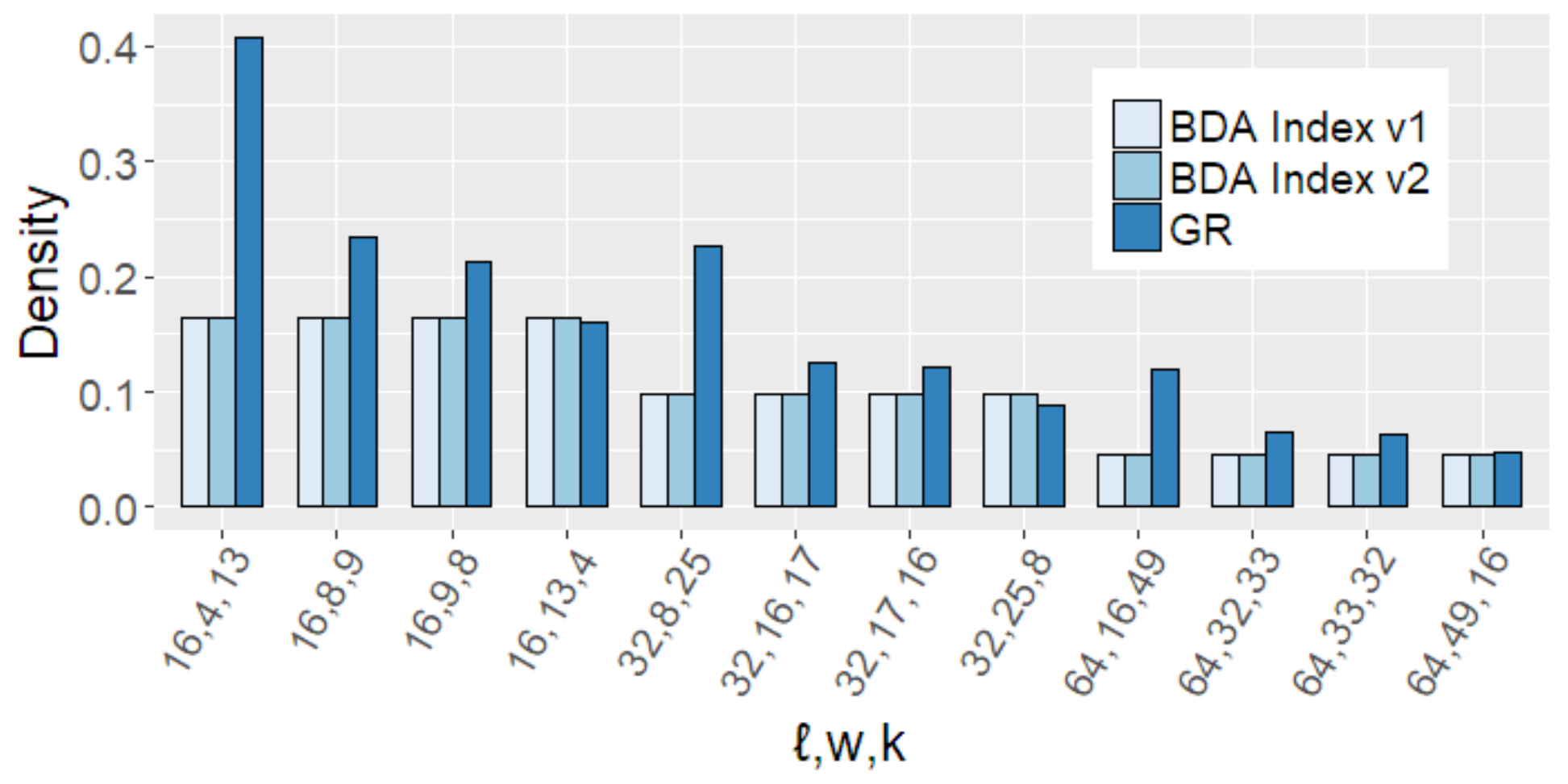}
     \caption{SOURCES}
     
     \end{subfigure}
     \caption{Density vs.~$w,k$ for $\ell=w+k-1$ and the datasets of Table~\ref{tab:data}.}\label{fig:experiment2-density}
     \end{figure}
    
    \begin{figure}[t]
     \begin{subfigure}[b]{0.453\textwidth}
     \includegraphics[width=1\linewidth]{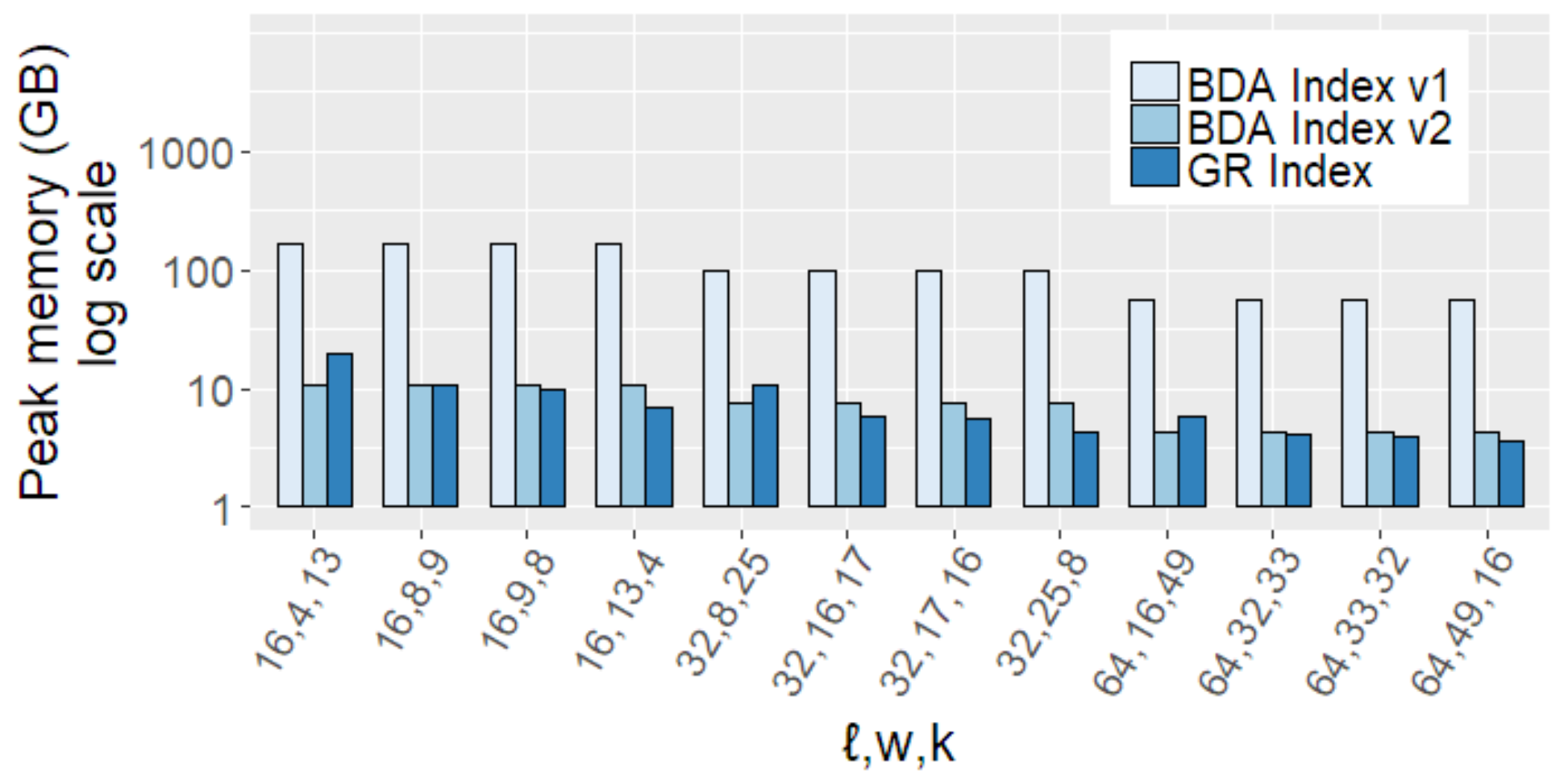}
     \caption{DNA}
     \end{subfigure}\hspace{+2mm}
     \begin{subfigure}[b]{0.453\textwidth}
     \includegraphics[width=1\linewidth]{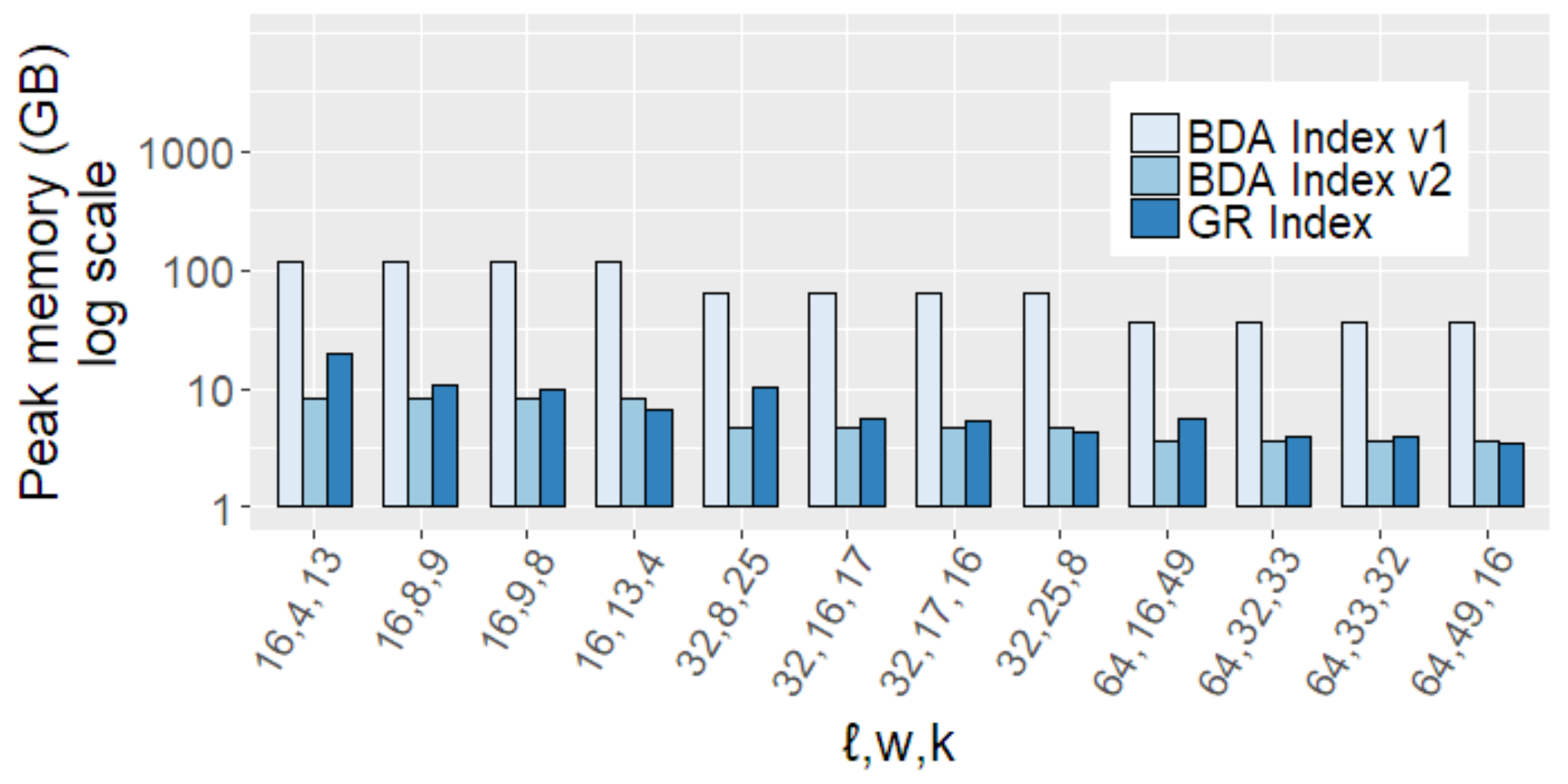}
     \caption{\DBLP}
     \end{subfigure}\\
\begin{subfigure}[b]{0.453\textwidth}
     \includegraphics[width=1\linewidth]{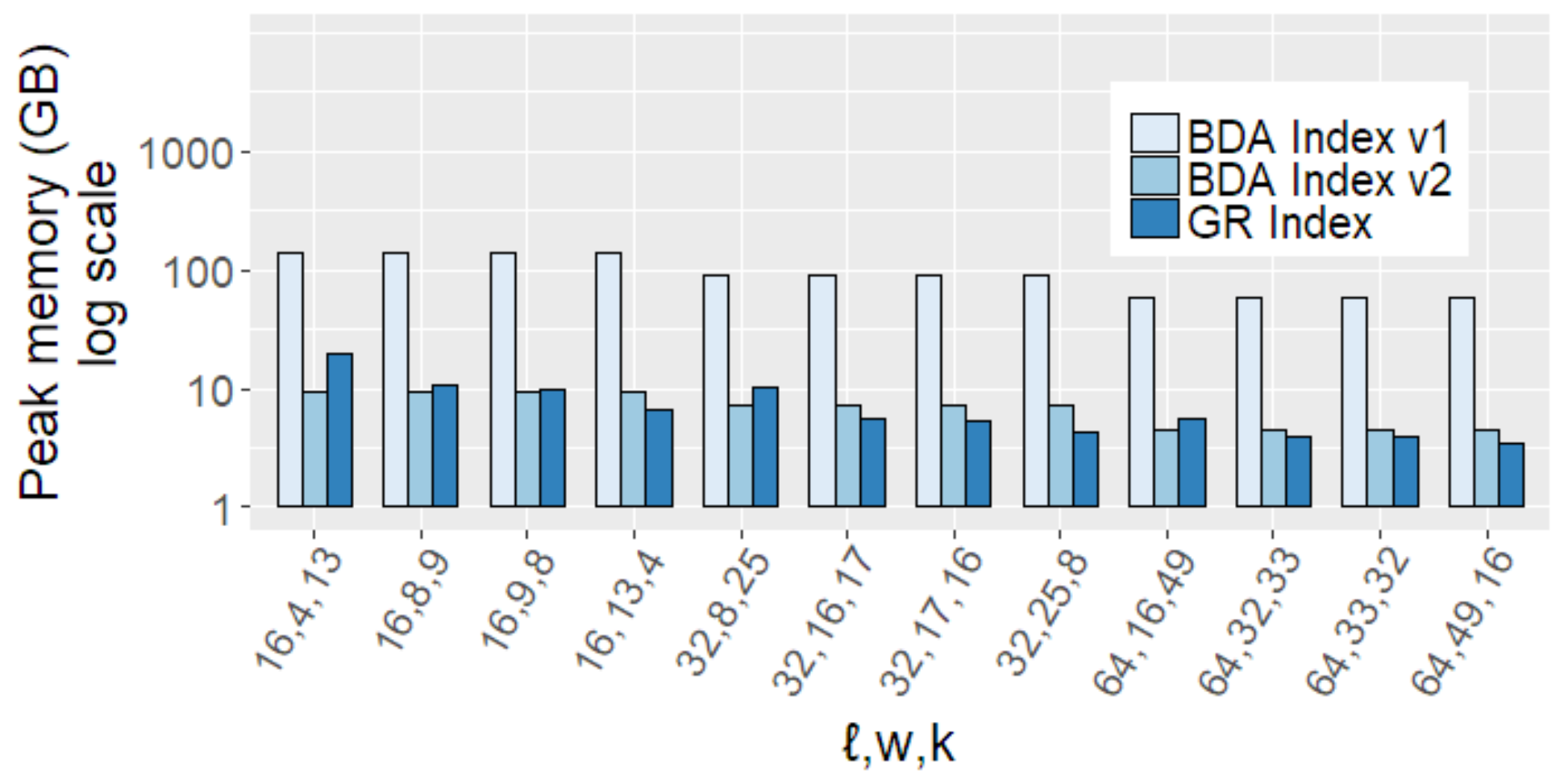}
     \caption{ENGLISH}
     \end{subfigure}\hspace{.2cm}
     \begin{subfigure}[b]{0.453\textwidth}
     \includegraphics[width=1\linewidth]{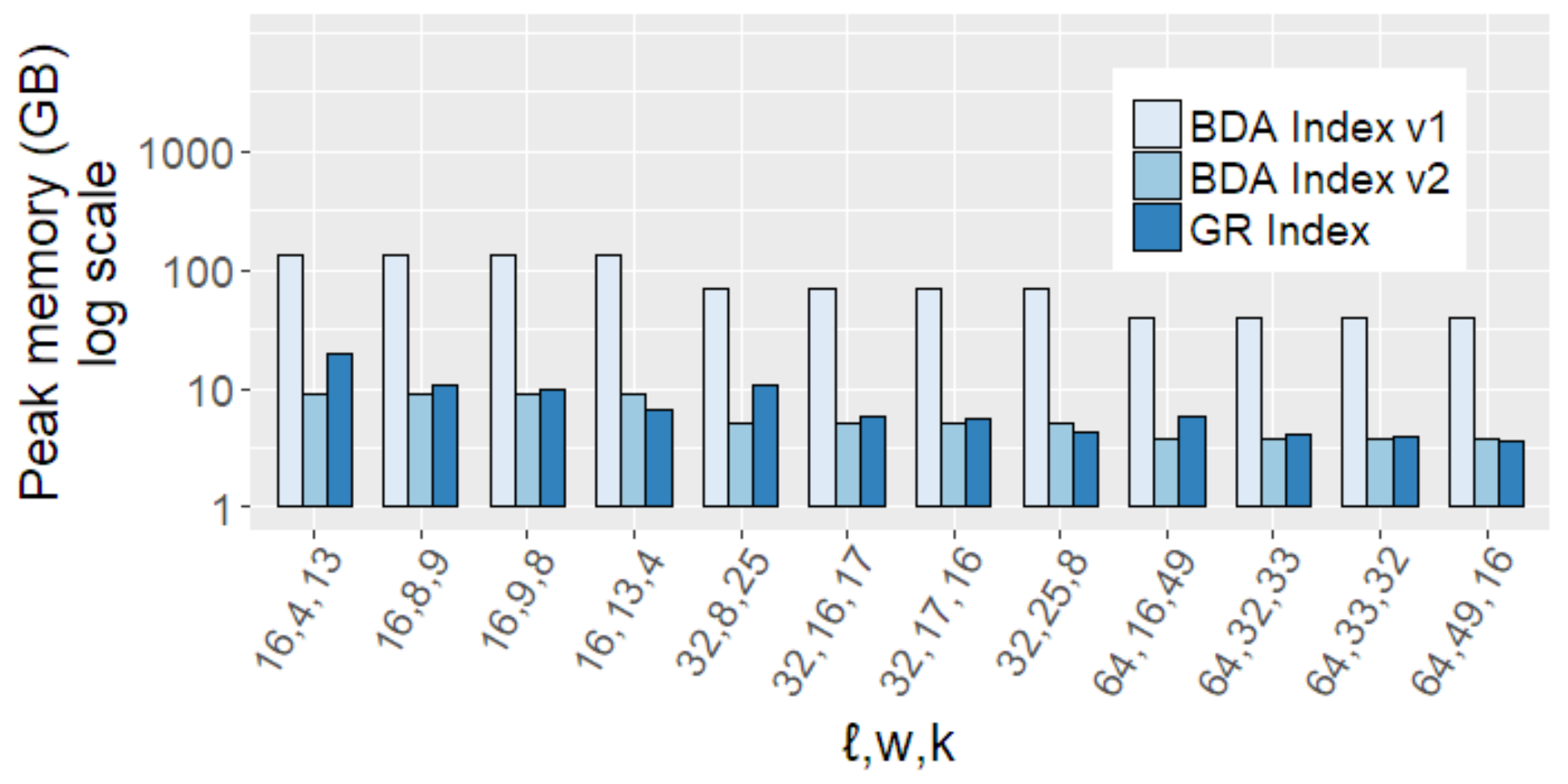}
     \caption{PROTEINS}
     \end{subfigure}\\\centering
     \begin{subfigure}[b]{0.453\textwidth}
     \includegraphics[width=1\linewidth]{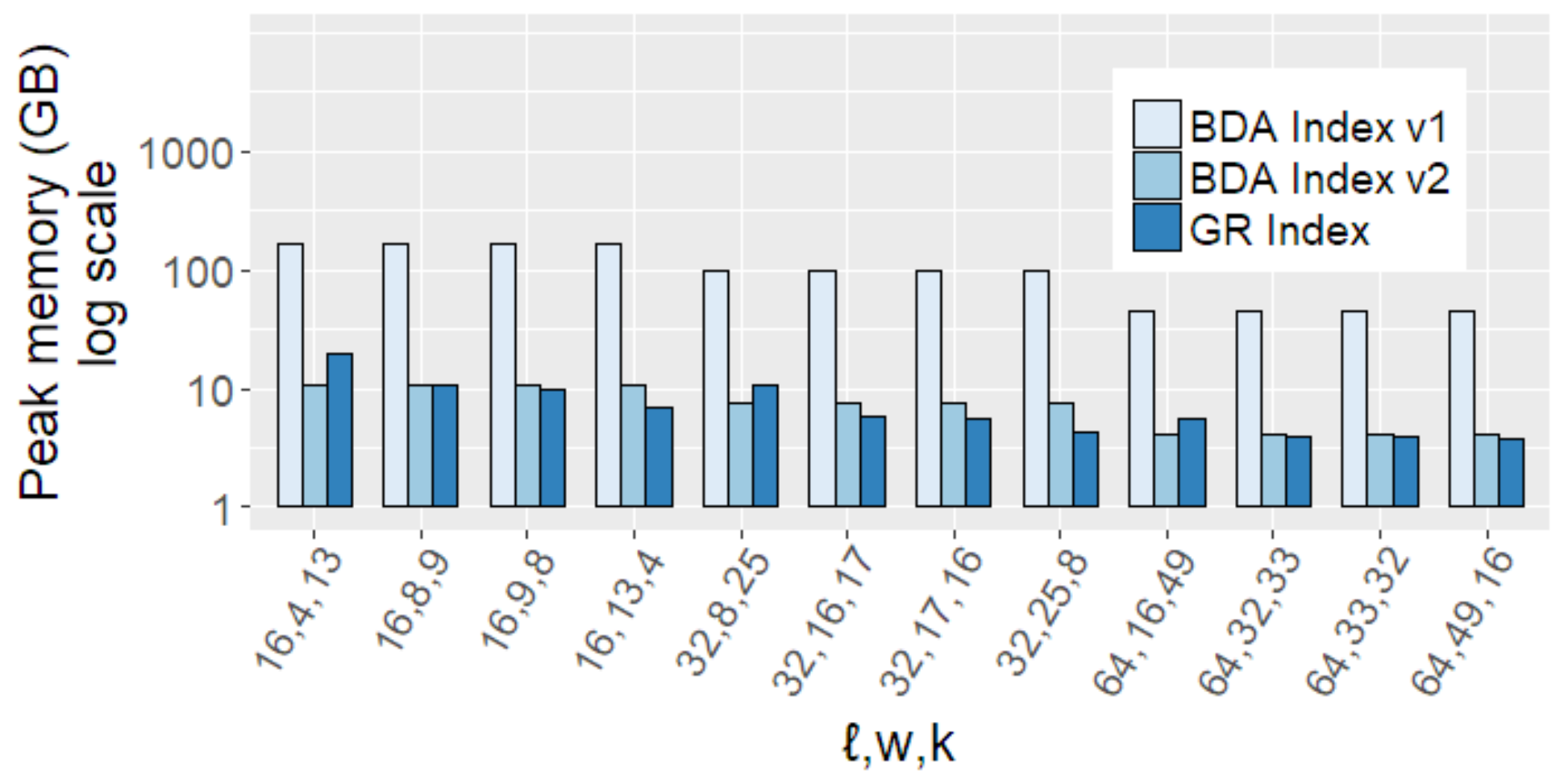}
     \caption{SOURCES}
     \end{subfigure}
     \caption{Peak memory usage (GB) vs.~$w,k$ for $\ell=w+k-1$ and the datasets of Table~\ref{tab:data}.}~\label{fig:experiment2-mem}
\end{figure}

\paragraph{Discussion.}~The proposed $\mathcal{I}_{\ell}(T)$ index, which is based on bd-anchors, has the following attributes:
\begin{enumerate}
    \item \textbf{Construction}: $\G_{\ell}(T)$ is constructed in $\cO(n)$ worst-case time and $\mathcal{I}_{\ell}(T)$ is constructed in $\cO(n +|\G_{\ell}(T)|\sqrt{\log(|\G_{\ell}(T)|)})$ worst-case time. These time complexities are near-linear in $n$ and do not depend on the alphabet $\Sigma$ as long as $|\Sigma|=n^{\cO(1)}$, which is true for virtually any real scenario.
    \item \textbf{Index Size}: By Theorem~\ref{the:main}, $\mathcal{I}_{\ell}(T)$ can occupy $\cO(|\G_{\ell}(T)|)$ space. By Lemma~\ref{lem:reduced}, the size of $\G_{\ell}(T)$ is $\cO(n/\ell)$ in expectation and so $\mathcal{I}_{\ell}(T)$ can also be of size $\cO(n/\ell)$. In practice this depends on $T$ and on the implementation of the 2D range reporting data structure.
    \item \textbf{Querying}: The $\mathcal{I}_{\ell}(T)$ index answers on-line pattern searches in near-optimal time. 
    \item \textbf{Flexibility}: Note that one would have to \emph{reconstruct} a (hash-based) index, which indexes the set of $(w,k)$-minimizers, to increase specificity or sensitivity: increasing $k$ increases the specificity and decreases the sensitivity. Our $\mathcal{I}_{\ell}(T)$ index, conceptually truncated at string depth $k$, is essentially an index based on $(w,k)$-minimizers, which additionally wrap around.  We can thus increase \emph{specificity} by considering larger $\alpha,\beta$ values or increase \emph{sensitivity} by considering smaller $\alpha,\beta$ values. This effect can be realized \emph{without reconstructing} our $\mathcal{I}_{\ell}(T)$ index: we just adapt $\alpha$ and $\beta$ upon querying accordingly. 
\end{enumerate}

\section{Similarity Search under Edit Distance}\label{sec:edit}

We show how bd-anchors can be applied to speed up similarity search under edit distance. This is a fundamental problem with myriad applications in bioinformatics, databases, data mining, and information retrieval. It has thus been studied extensively in the literature both from a theoretical and a practical point of view~\cite{DBLP:conf/vldb/KahveciS01,DBLP:conf/sigmod/ChaudhuriGGM03,DBLP:conf/stoc/ColeGL04,DBLP:conf/vldb/LiWY07,DBLP:conf/aaai/YangYK10,DBLP:conf/sigmod/ZhangHOS10,DBLP:journals/tods/Qin0XLLW13,DBLP:journals/pvldb/WangDTZ13,DBLP:conf/sigmod/WangLF12,DBLP:conf/sigmod/DengLF14,DBLP:journals/tkde/HuLBFWGX16,DBLP:journals/vldb/YuWLZDF17,DBLP:conf/kdd/Zhang020}. Let $\mathcal{D}$ be a collection of strings called \emph{dictionary}. We focus, in particular, on indexing $\mathcal{D}$ for answering the following type of top-$K$ queries: Given a query string $Q$ 
and an integer $K$, return $K$ strings from the dictionary that are closest to $Q$ with respect to edit distance. We follow a typical seed-chain-align approach as used by several bioinformatics  applications~\cite{BLAST,mummer,10.1093/bioinformatics/btw152,DBLP:journals/bioinformatics/Li18}. The main new ingredients we inject, with respect to this classic approach, is that we use: (1) bd-anchors as seeds; and (2) $\mathcal{I}_{\ell}$ to index the dictionary $\mathcal{D}$, for some integer parameter $\ell>0$.

\paragraph{Construction.}~We require an integer parameter $\ell>0$ defining the order of the bd-anchors. We set $T=S_1\ldots S_{|\mathcal{D}|}$, where $S_i\in \mathcal{D}$, compute the bd-anchors of order $\ell$ of $T$, and construct the $\mathcal{I}_{\ell}(T)$ index (see Section~\ref{sec:index}) using the bd-anchors. 

\paragraph{Querying.}~We require two parameters $\tau\geq 0$ and $\delta\geq 0$. The former parameter controls the sensitivity of our filtering step (Step 2 below); and the latter one controls the sensitivity of our verification step (Step 3 below). Both parameters trade accuracy for speed.
\begin{enumerate}
    \item For each query string $Q$, we compute the bd-anchors of order $\ell$. For every bd-anchor $j_Q$, we take an arbitrary fragment (e.g.~the leftmost) of length $\ell$ anchored at $j_Q$ as the \emph{seed}. Let this fragment start at position $i_Q$. This implies a value for $\alpha$ and $\beta$, with $\alpha+\beta=\ell+1$; specifically for $Q[i_Q\dd i_Q+\ell-1]$ we have $Q[i_Q\dd j_Q]=Q[j_Q-\alpha+1\dd j_Q]$ and $Q[j_Q\dd i_Q+\ell-1]=Q[j_Q\dd j_Q+\beta-1]$.
    For every bd-anchor $j_Q$, we query $\overleftarrow{Q[j_Q-\alpha+1\dd j_Q]}$ in $\T^{L}_{\ell}(T)$ and $Q[j_Q\dd j_Q+\beta-1]$ in $\T^{R}_{\ell}(T)$ and collect
    all $(\alpha,\beta)$-hits.
    \item Let $\tau\geq 0$ be an input parameter and let $L_{Q,S}=(q_1,s_1),\ldots,(q_k,s_h)$ be the list of all $(\alpha,\beta)$-hits between the queried fragments of string $Q$ and fragments of a string $S\in \mathcal{D}$. 
    If $h<\tau$, we consider string $S$ as not found. 
    The intuition here is that if $Q$ and $S$ are sufficiently close with respect to edit distance, they would have a relatively long $L_{Q,S}$~\cite{mummer}.
    If $h\geq \tau$, we sort the elements of $L_{Q,S}$ with respect to their first component. 
    (This comes for free because we process $Q$ from left to right.)
    We then compute a \emph{longest increasing subsequence} (LIS) in $L_{Q,S}$ with respect to the second component, which \emph{chains} the $(\alpha,\beta)$-hits, in $\cO(h\log h)$ time~\cite{schensted_1961} per $L_{Q,S}$ list. We use the LIS of $L_{Q,S}$ to \emph{estimate} the  \emph{identity score} (total number of matching letters in a fixed alignment) for $Q$ and $S$, which we denote by $E_{Q,S}$, based on the occurrences of the $(\alpha,\beta)$-hits in the LIS. 
    \item  Let $\delta\geq 0$ be an input parameter and let $E_K$ be the $K$th largest estimated identity score. We extract, as candidates, the ones whose estimated identity score is at least $E_K-\delta$. 
    For every candidate string $S$, we close the gaps between the occurrences of the $(\alpha,\beta)$-hits in the LIS using dynamic programming~\cite{SovietUnion}, thus computing an \emph{upper bound} on the edit distance between $Q$ and $S$ (UB score).
    In particular, closing the gaps consists in summing up the exact edit distance for all pairs of fragments (one from $S$ and one from $Q$) that lie in between the $(\alpha,\beta)$-hits.
    We return $K$ strings from the list of candidates with the lowest UB score. If $\delta=0$, we return $K$ strings with the highest $E_{Q,S}$ score.
\end{enumerate}

\paragraph{Index Evaluation.}~We compared our algorithm, called \textsf{BDA Search}, to \textsf{Min Search}, the state-of-the-art tool for top-$K$ similarity search under edit distance proposed by Zhang and Zhang  in~\cite{DBLP:conf/kdd/Zhang020}. The main concept used in \textsf{Min Search} is the rank of a letter in a string, defined as the size of the neighborhood of the string in which the letter has the minimum hash value. Based on this concept, \textsf{Min Search} partitions each string in the dictionary $\mathcal{D}$ into a hierarchy of substrings and then builds an index comprised of a set of hash tables, so that strings having common substrings and thus small edit distance are grouped into the same hash table. To find the top-$K$ closest strings to a query string, \textsf{Min Search} partitions the query string based on the ranks of its letters and then traverses the hash tables comprising the index. Thanks to the index and the use of several filtering tricks, \textsf{Min Search} is at least one order of magnitude faster with respect to query time than popular alternatives~\cite{DBLP:journals/vldb/YuWLZDF17,DBLP:conf/sigmod/ZhangHOS10,icdecompetitor3}.

We implemented two versions of \textsf{BDA Search}: \textsf{BDA Search v1} which is based on \textsf{BDA Index v1}; and \textsf{BDA Search v2} which is based on \textsf{BDA Index v2}. For \textsf{Min Search}, we used the \texttt{C++} implementation from~\url{https://github.com/kedayuge/Search}.  

We constructed  synthetic datasets, referred to as \SYN, in a way that enables us to study the impact of different parameters and efficiently identify the ground truth (top-$K$ closest strings to a query string with respect to edit distance). Specifically, we first generated $50$ query strings and then constructed a cluster of $K$ strings around each query string. To generate the query strings, we started  from an arbitrary string $Q$ of length $|Q|=1000$ from a real dataset of protein sequences, used in~\cite{DBLP:conf/kdd/Zhang020}, and generated a string $Q'$ that is at edit distance $e$ from $Q$, by performing $e$ edit distance operations, each with equal probability. Then, we treated $Q'$ as $Q$ and repeated the process to generate the next query string.   To create the clusters, we first added each query string into an initially empty cluster and then added $K-1$ strings, each at edit distance at most $e'<e$ from the query string. The strings were generated by performing at most $e'$ edit distance operations, each with equal probability. Thus, each cluster contains the top-$K$ closest strings to the query string of the cluster. We used $K\in\{5,10,15,20,25\}$, $d=\frac{e}{|Q|}\in \{0.1,0.15,0.2,0.25,0.3\}$, and $d'=\frac{e'}{|Q|}=d-0.05$.  We evaluated query answering accuracy using the F1 score~\cite{manning}, expressed as the harmonic mean of precision and recall\footnote{Precision is the ratio between the number of returned strings that are among the top-$K$ closest strings to a query string and the number of all returned strings. Recall is the ratio between the number of returned strings that are among the top-$K$ closest strings to a query string and $K$. Since all tested algorithms return $K$ strings, the F1 score in our experiments is equal to precision and equal to recall.}.
For \textsf{BDA Search}, we report results for $\tau=0$ (full sensitivity during filtering) and $\delta=0$ (no sensitivity during verification), as it was empirically determined to be a reasonable trade-off between accuracy and speed. 
For \textsf{Min Search}, we report results using its default parameters from~\cite{DBLP:conf/kdd/Zhang020}.

We plot the F1 scores and average query time in Figures~\ref{fig:f1} and~\ref{fig:syn-time}, respectively. 
All methods achieved  almost perfect accuracy, in all tested cases. \textsf{BDA Search} slightly outperformed \textsf{Min Search} (by up to $1.1\%$), remaining accurate even for large $\ell$; the changes to F1 score for \textsf{Min Search} as $\ell$ varies are because the underlying method is randomized. However, both versions of \textsf{BDA Search} were \emph{more than one order of magnitude faster} than \textsf{Min Search} on average (over all results of Figure~\ref{fig:syn-time}), with \textsf{BDA Search v1} being 2.9 times slower than \textsf{BDA Search v2} on average, due to the inefficiency of the range tree implementation of CGAL. Furthermore, both versions of \textsf{BDA Search} scaled better with respect to $K$. For example, the average query time for \textsf{BDA Search v1} became 2 times larger when $K$ increased from $5$ to $25$ (on average over $\ell$ values), while that for \textsf{Min Search} became 5.4 times larger on average. The reason is that verification in \textsf{Min Search}, which increases the accuracy of this method, becomes increasingly expensive as $K$ gets larger. The peak memory usage for these experiments is reported in Figure~\ref{fig:edit-syn-mem}. Although \textsf{Min  Search} outperforms \textsf{BDA Search} in terms of memory usage,  \textsf{BDA Search v2} still required a very small amount of memory (less than 1GB).  \textsf{BDA Search v1} required more memory for the reasons mentioned in Section~\ref{sec:index}. 

\begin{figure}[!t]
     \begin{subfigure}[b]{0.45\textwidth}
     \includegraphics[width=1\linewidth]{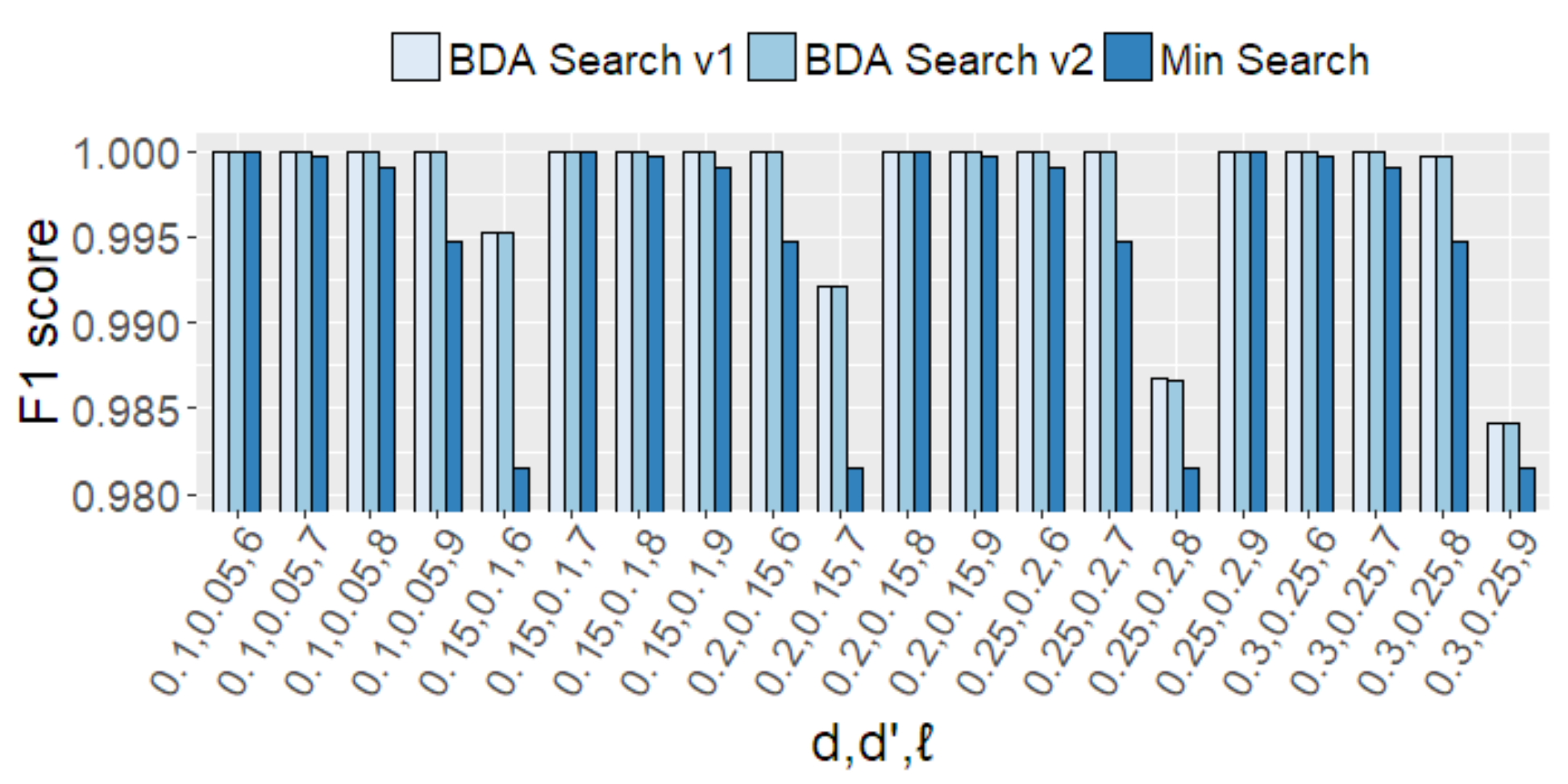}
     \caption{SYN}
     \end{subfigure}\hspace{+2mm}
     \begin{subfigure}[b]{0.45\textwidth}
     \includegraphics[width=1\linewidth]{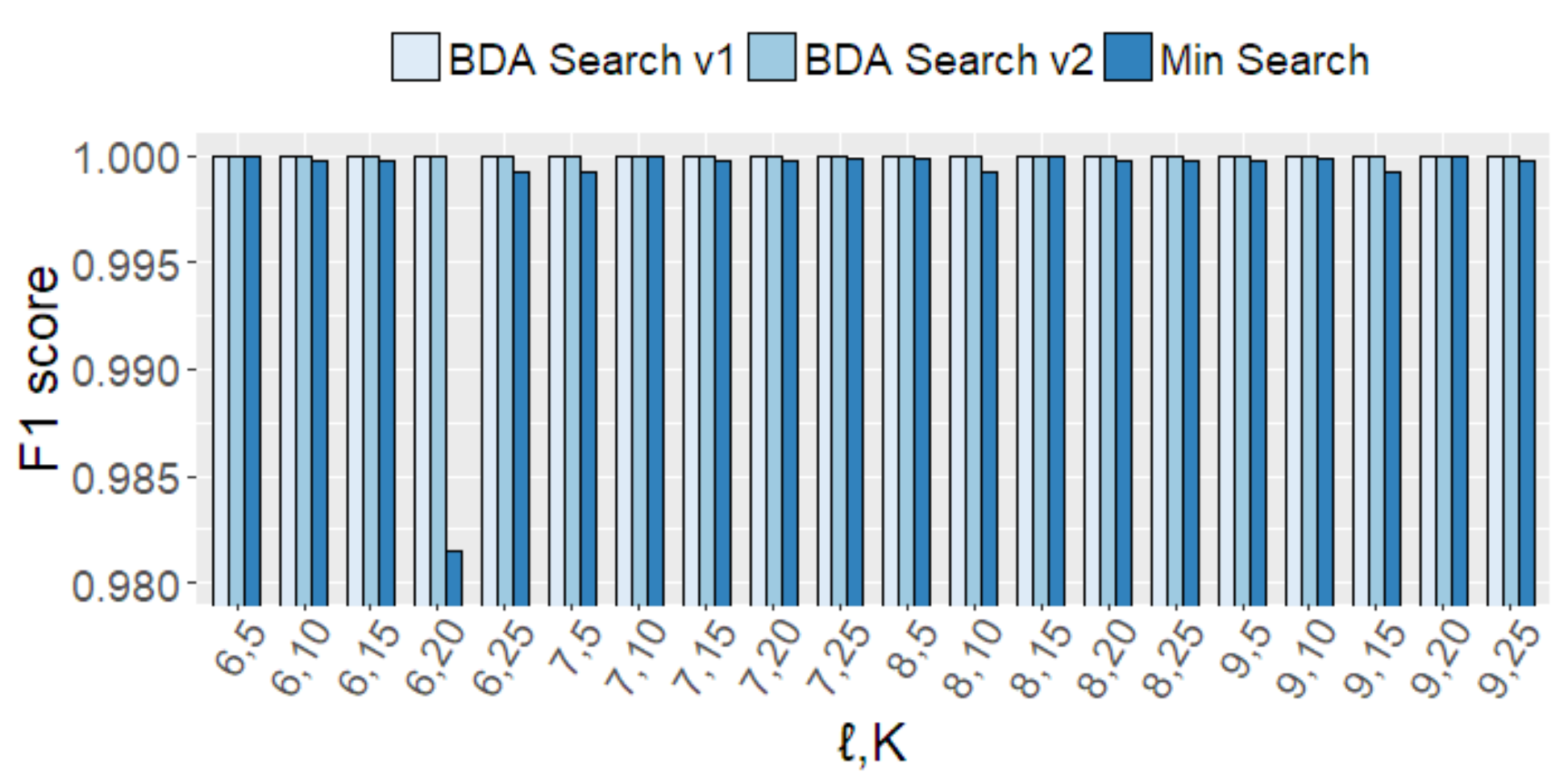}
     \caption{SYN}
     \end{subfigure}
     \caption{F1 score vs.~(a) $d$, $d'$, $\ell$, for $K=20$, and (b) $\ell$, $K$, for $d=0.15$ and $d'=0.1$.}\label{fig:f1}
\end{figure}

\begin{figure}[!t]
     \begin{subfigure}[b]{0.45\textwidth}
     \includegraphics[width=1\linewidth]{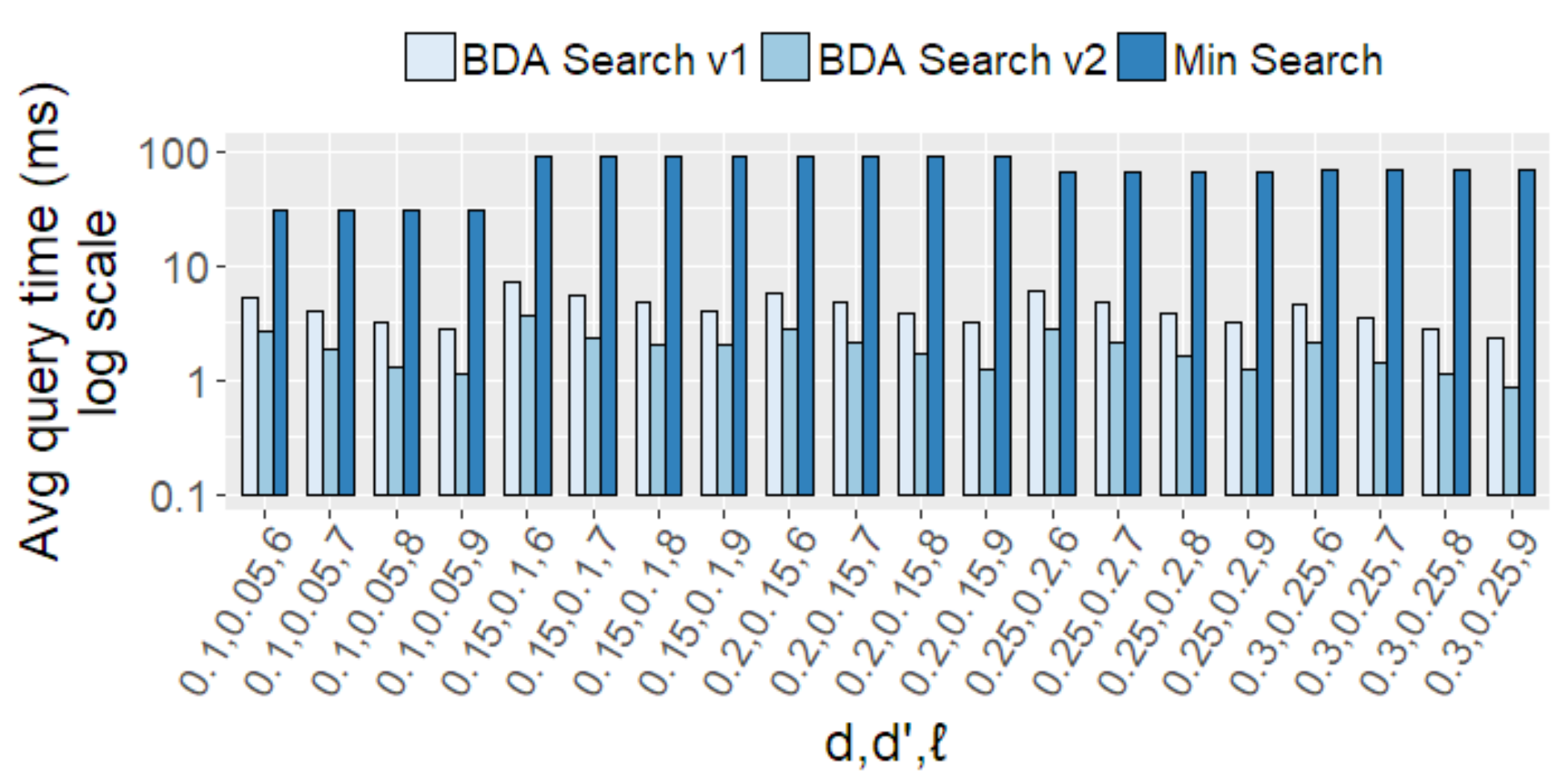}
     \caption{SYN}\label{fig:syn-time-vsd}
     \end{subfigure}\hspace{+2mm}
     \begin{subfigure}[b]{0.45\textwidth}
     \includegraphics[width=1\linewidth]{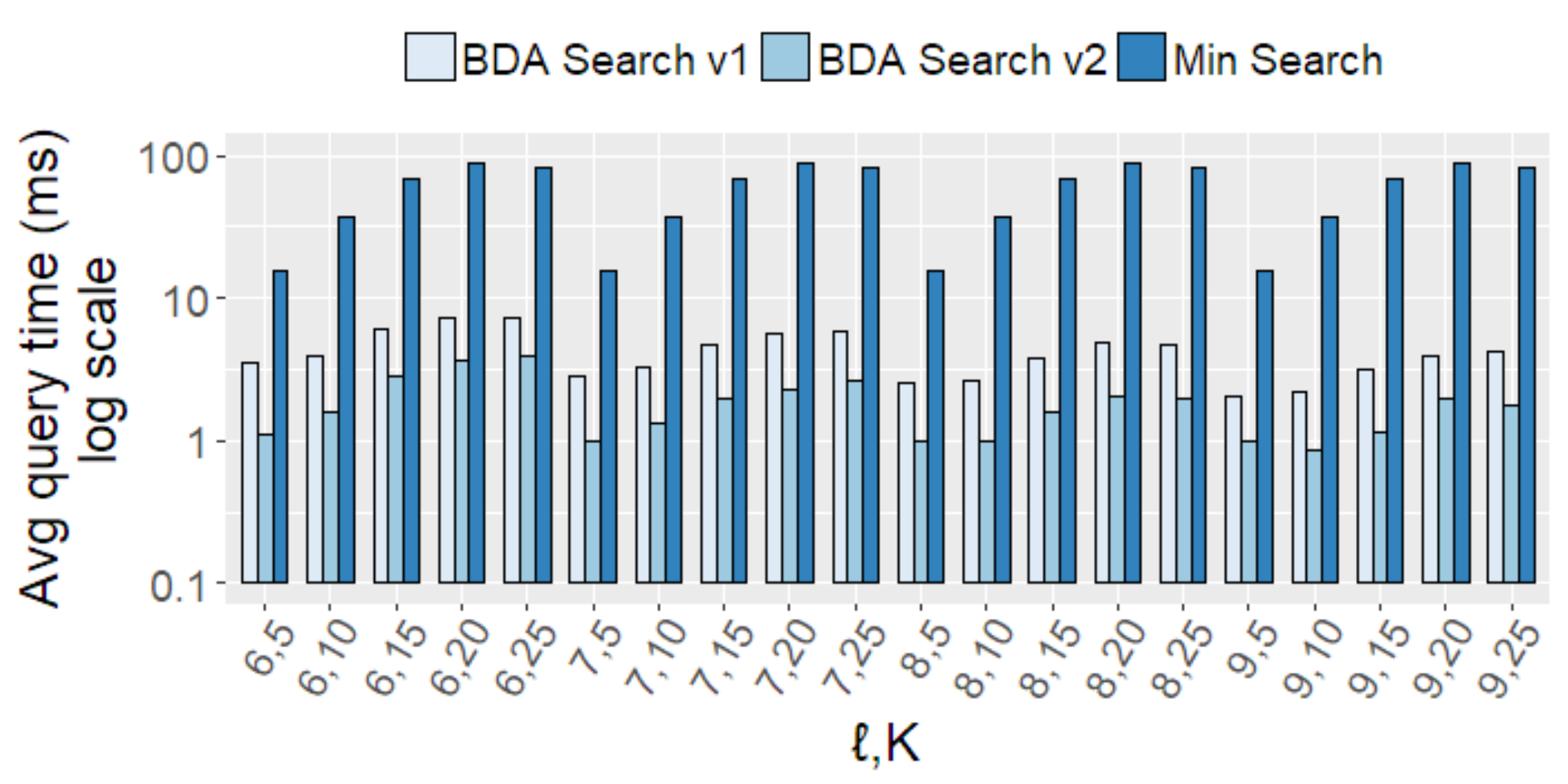}
     \caption{SYN}
     \end{subfigure}\label{fig:syn-time-vsK}
\caption{Average query time (ms) vs.~(a) $d$, $d'$, $\ell$, for $K=20$, and (b) $\ell$, $K$, for $d=0.15$ and $d'=0.1$.}\label{fig:syn-time}
\end{figure}

\begin{figure}[!t]
\begin{subfigure}[b]{0.453\textwidth}
     \includegraphics[width=1\linewidth]{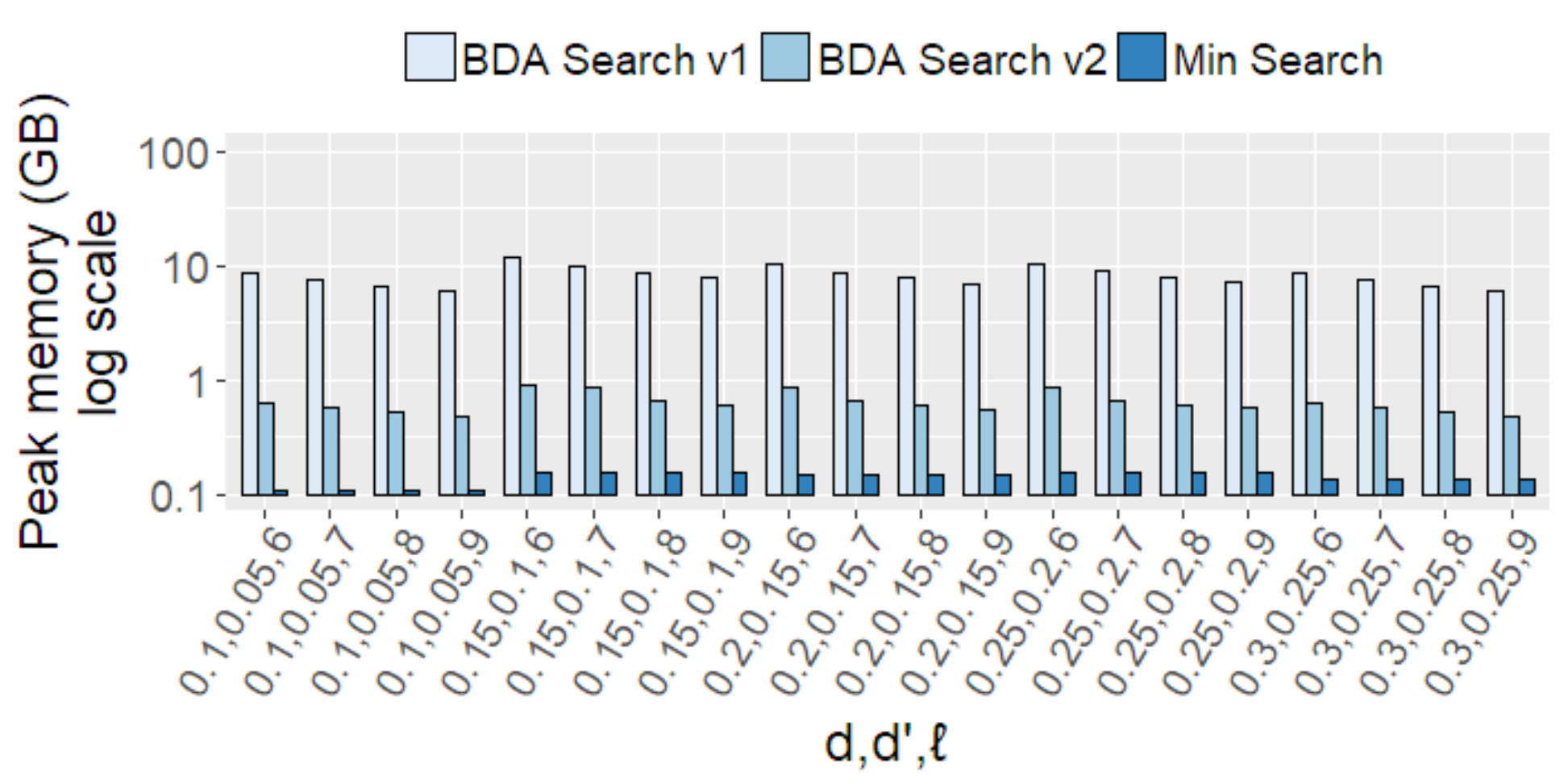}
     \caption{SYN}
     \end{subfigure}\hspace{+1mm}
     \begin{subfigure}[b]{0.453\textwidth}
     \includegraphics[width=1\linewidth]{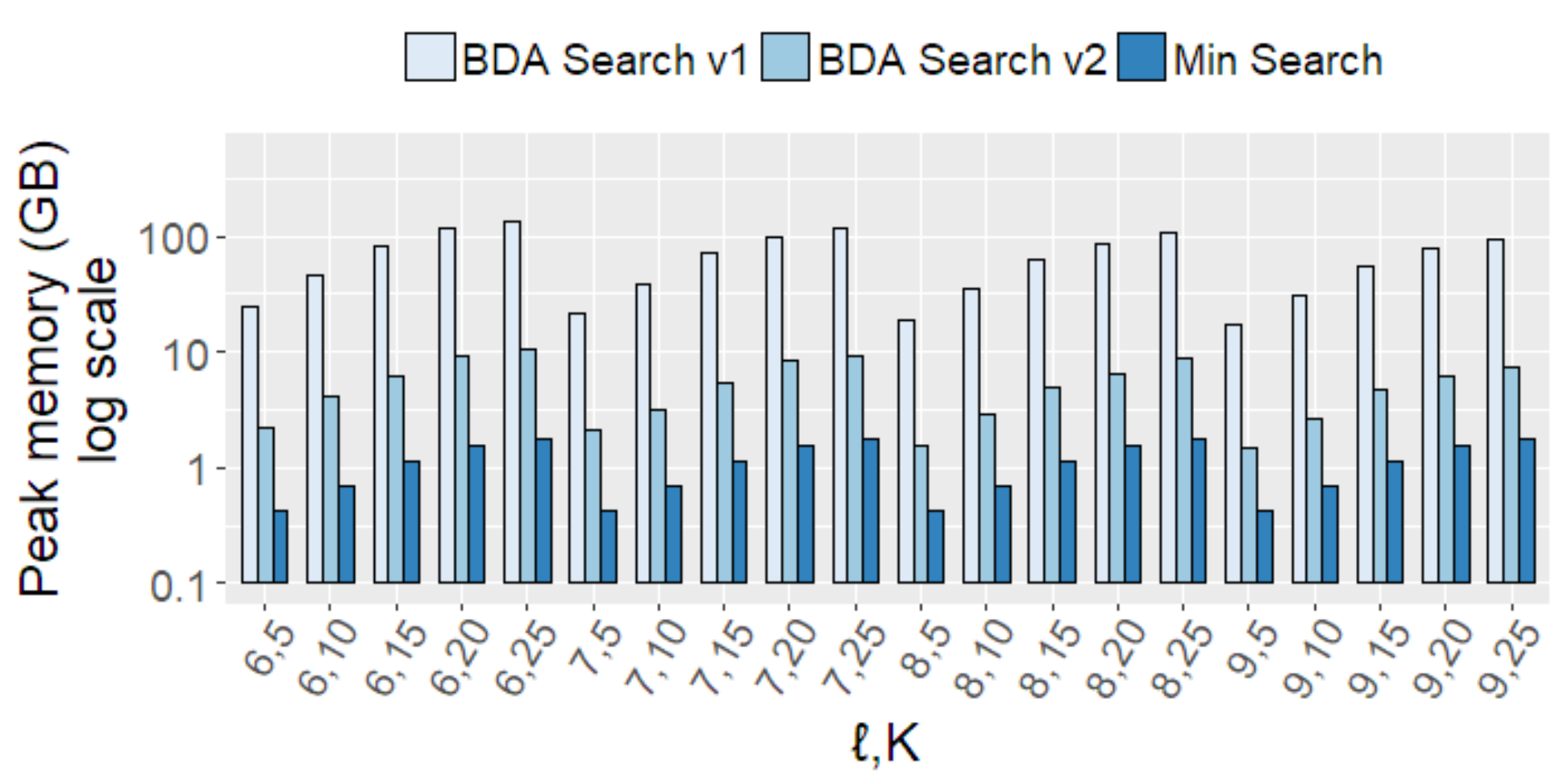}
     \caption{SYN}
     \end{subfigure}
     \caption{Peak memory usage (GB) vs.~(a) $d$, $d'$, $\ell$, for $K=20$, and (b) $\ell$, $K$, for $d=0.15$ and $d'=0.1$.}\label{fig:edit-syn-mem}
\end{figure}

\paragraph{Discussion.}~\textsf{BDA Search} outperforms \textsf{Min Search} in accuracy while being more than one order of magnitude faster in query time. These results are very encouraging because the efficiency of \textsf{BDA Search} is entirely due to injecting bd-anchors and not due to any further filtering tricks such as those employed by \textsf{Min Search}. \textsf{Min Search} clearly outperforms \textsf{BDA Search} in memory usage, albeit the memory usage of \textsf{BDA Search v2} is still quite modest. We defer an experimental evaluation using real datasets to the journal version of our work.

\section{Other Works on Improving Minimizers}\label{sec:related}
Although every sampling mechanism based on minimizers primarily aims at satisfying Properties 1 and 2, different mechanisms employ total orders that lead to substantially different total numbers of selected minimizers. Thus,  research on minimizers has focused on determining total orders which lead to the lowest possible density (recall that the density is defined as the number of selected length-$k$ substrings over the length of the input string). In fact, much of the literature focuses on the \emph{average case}~\cite{DBLP:conf/wabi/OrensteinPMSK16,DBLP:journals/bioinformatics/MarcaisPBOSK17,DBLP:journals/bioinformatics/MarcaisDK18,DBLP:conf/recomb/EkimBO20,DBLP:journals/bioinformatics/ZhengKM20}; namely, the lowest expected density when the input string is random. In practice, many works use a ``random minimizer'' where the order is defined by choosing a permutation of all the length-$k$ strings at random (e.g.~by using a hash function, such as the Karp-Rabin fingerprints~\cite{DBLP:journals/ibmrd/KarpR87}, on the length-$k$ strings). Such a randomized mechanism has the benefit of being easy to implement and providing good expected performance in practice. 

\paragraph{Minimizers and Universal Hitting Sets.}~A \emph{universal hitting set} (UHS) is an unavoidable set of length-$k$ strings, i.e., it is a set of length-$k$ strings that ``hits'' every $(w+k-1)$-long fragment of every possible string. The theory of universal hitting sets~\cite{DBLP:conf/wabi/OrensteinPMSK16,DBLP:journals/bioinformatics/MarcaisDK18,DBLP:conf/stoc/KempaK19,DBLP:conf/recomb/ZhengKM20} plays an important role in the current theory for minimizers with low density on average. In particular, if a UHS has small size, it generates minimizers with a provable upper-bound on their density. However, UHSs are less useful in the string-specific case for two reasons~\cite{Zheng2021}: (1) the requirement that a UHS has to hit every $(w+k-1)$-long fragment of every possible string is too strong; and (2) UHSs are too large to provide a meaningful upper-bound on the density in the string-specific case. 
Therefore, since in many practical scenarios the input string is known and does not change frequently, we try to optimize the density for one particular string instead of optimizing the average density over a random input. 

\paragraph{String-Specific Minimizers.}~In the string-specific case, minimizers sampling mechanisms may employ frequency-based orders~\cite{DBLP:journals/bioinformatics/ChikhiLM16,DBLP:journals/bioinformatics/JainRZCWKP20}. In these orders, length-$k$ strings occurring  less  frequently  in  the string compare  less  than the ones occurring more frequently.  The intuition~\cite{Zheng2021} is to obtain a sparse sampling by selecting infrequent length-$k$ strings which should be spread apart in the string. However, there is no theoretical guarantee that a frequency-based order gives low density minimizers (there are many  counter-examples). Furthermore, frequency-based orders do not always give minimizers with lower density in practice. For instance,~the two-tier classification (very frequent vs.~less frequent length-$k$ strings) in the work of \cite{DBLP:journals/bioinformatics/JainRZCWKP20} outperforms an order that strictly follows frequency of occurrence. 

A different approach to constructing string-specific minimizers  is to  start  from  a UHS  and  to  remove  elements from it, as  long  as  it  still  hits  every $(w+k-1)$-long  fragment  of the input  string~\cite{DBLP:conf/bcb/DeBlasioGKM19}. Since this approach starts with a UHS that is not related to the string, the improvement in density may not be significant~\cite{Zheng2021}. Additionally, current methods~\cite{DBLP:conf/recomb/EkimBO20} employing this approach are computationally limited to using $k\leq 16$, as the size of the UHS increases exponentially with $k$. Using such small $k$ values may not be appropriate in some applications.

\paragraph{Other Improvements.}~When $k\approx w$, minimizers with expected density of $1.67/w+o(1/w)$ on a random string can be constructed using the approach  of~\cite{DBLP:journals/bioinformatics/ZhengKM20}. Such minimizers have  guaranteed expected density less than $2/(w+1)$ and work for infinitely many $w$ and $k$. The approach of~\cite{DBLP:journals/bioinformatics/ZhengKM20} also does not require the use of expensive heuristics to precompute and store a large set of length-$k$ strings, unlike some methods~\cite{DBLP:conf/wabi/OrensteinPMSK16,DBLP:conf/bcb/DeBlasioGKM19,DBLP:conf/recomb/EkimBO20} with low density in practice.  

The notion of \emph{polar set}, which can be seen as complementary to that of UHS, was recently introduced in~\cite{Zheng2021}. While a UHS is a set of length-$k$ strings that intersect with  every $(w+k-1)$-long  fragment at least once, a polar set is a set of length-$k$ strings that intersect with any fragment at most once.  The  construction of a polar set builds upon sets of length-$k$ strings that are sparse in the input  string. Thus, the minimizers derived from these polar sets have provably tight bounds on their density. Unfortunately, computing optimal polar sets is NP-hard, as shown in~\cite{Zheng2021}. Thus, the work of~\cite{Zheng2021} also proposed a heuristic for computing feasible ``good enough'' polar sets. A main disadvantage of this approach is that when each length-$k$ string occurs frequently in the input string, it becomes hard to select many length-$k$ strings without violating the polar set condition.

\bibliography{references.bib}

\end{document}